\documentclass[11pt]{article}
\usepackage{amsthm,amssymb,amsmath,fullpage,bm}
\usepackage[usenames]{color}
\usepackage[colorlinks=true]{hyperref}
\usepackage{algorithm}
\usepackage{algorithmicx}
\usepackage[noend]{algpseudocode}

\newtheorem{theorem}{Theorem}[section]
\newtheorem{lemma}[theorem]{Lemma}
\newtheorem{proposition}[theorem]{Proposition}

\newtheorem{corollary}[theorem]{Corollary}

\newtheorem{examplecore}[theorem]{Example}

\newenvironment{example}
  {\begin{examplecore}\rm}
  {\hfill $\Box$ \end{examplecore}}

\newtheorem{definition}[theorem]{Definition}

\newcommand{\Var}{\mathop{\mathbf{Var}}}
\newcommand{\E}{\mathop{\mathbf{E}}}
\newcommand{\maj}{\mathrm{maj}}
\newcommand{\meet}{\wedge}
\newcommand{\join}{\vee}
\newcommand{\Alg}{\mathrm{Alg}}
\newcommand{\Pol}{\mathrm{Pol}}

\newcommand{\CSP}{\mathrm{CSP}}
\newcommand{\eCSP}{\exists\mathrm{CSP}}
\newcommand{\bw}{{\bm w}}
\newcommand{\bbA}{\mathbb{A}}
\newcommand{\bbB}{\mathbb{B}}
\newcommand{\bbC}{\mathbb{C}}
\newcommand{\bbF}{\mathbb{F}}
\newcommand{\bbL}{\mathbb{L}}

\newcommand{\bbR}{\mathbb{R}}

\newcommand{\bfA}{\mathbf{A}}
\newcommand{\bfB}{\mathbf{B}}
\newcommand{\bfP}{\mathbf{P}}
\newcommand{\bfR}{\mathbf{R}}
\newcommand{\caA}{\mathcal{A}}

\newcommand{\caC}{\mathcal{C}}

\newcommand{\caI}{\mathcal{I}}

\newcommand{\caP}{\mathcal{P}}
\newcommand{\caS}{\mathcal{S}}

\newcommand{\caU}{\mathcal{U}}
\newcommand{\caV}{\mathcal{V}}
\newcommand{\pr}{\mathrm{pr}}

\newcommand{\set}[1]{\{#1\}}

\newcommand{\dist}{\mathrm{dist}}

\begin{document}

\title{Testing Assignments to Constraint Satisfaction Problems}
\author{Hubie Chen\\
Univ. Pa\'is Vasco and IKERBASQUE \\
  E-20018 San Sebasti\'{a}n,
  Spain\\
  \texttt{hubie.chen@ehu.es}
  \and
  Matt Valeriote\thanks{Supported by a grant from the Natural Sciences and Engineering Research Council of Canada}\\
  McMaster University\\
  Hamilton, Canada\\
  \texttt{matt@math.mcmaster.ca}
  \and
  Yuichi Yoshida\thanks{Supported by JSPS Grant-in-Aid for Young Scientists (B) (No.~26730009), MEXT Grant-in-Aid for Scientific Research on Innovative Areas (24106003), and JST, ERATO,Kawarabayashi Large
 Graph Project.}\\
  National Institute of Informatics\\
  Chiyoda-ku, Tokyo 101-8430, Japan\\
  \emph{and}\\
  Preferred Infrastructure, Inc.\\
  Bunkyo-ku, Tokyo 113-0033, Japan\\
  \texttt{yyoshida@nii.ac.jp}
}
\date{}
\maketitle

\begin{abstract}
  For a finite relational structure $\bfA$, let $\CSP(\bfA)$ denote the CSP instances whose constraint relations are taken from $\bfA$.
  The resulting family of problems $\CSP(\bfA)$ has been considered heavily in a variety of computational contexts.  In this article, we consider this family from the perspective of property testing:
  given an instance of a CSP and query access to an assignment,
  one wants to decide whether the assignment satisfies the instance,
  or is far from so doing.
  While previous works on this scenario
  studied concrete templates or restricted classes of structures, this article presents comprehensive classification theorems.

  Our first contribution is a dichotomy theorem 
  completely characterizing
  the structures $\bfA$ such that $\CSP(\bfA)$ is constant-query testable:
  \begin{itemize}
  \itemsep=0pt
  \item If  $\bfA$ has a  majority polymorphism and a Maltsev polymorphism, then $\CSP(\bfA)$ is constant-query testable with one-sided error.
  \item Else, testing $\CSP(\bfA)$ requires a super-constant number of queries.
  \end{itemize}

  Let $\eCSP(\bfA)$ denote the extension of $\CSP(\bfA)$
  to instances which may include existentially quantified variables.
  Our second contribution is to classify all structures $\bfA$ in terms of the number of queries needed to test assignments to instances of $\eCSP(\bfA)$, with one-sided error.
  More specifically, we show the following trichotomy:
  \begin{itemize}
  \itemsep=0pt
  \item If $\bfA$ has a majority polymorphism and a Maltsev polymorphism, then $\eCSP(\bfA)$ is constant-query testable with one-sided error.
  \item Else, if $\bfA$ has a $(k+1)$-ary near-unanimity polymorphism for some $k \geq 2$, and no Maltsev polymorphism then $\eCSP(\bfA)$ is not constant-query testable (even with two-sided error) but is sublinear-query testable with one-sided error.
  \item Else, testing $\eCSP(\bfA)$ with one-sided error requires a linear number of queries.
  \end{itemize}
\end{abstract}

\thispagestyle{empty}
\setcounter{page}{0}

\newpage


\section{Introduction}\label{sec:intro}

\subsection{Background}

In property testing, the goal is to design algorithms that distinguish objects satisfying some predetermined property $P$ from objects that are far from satisfying $P$.
More specifically,
for $\epsilon,\delta > 0$, an algorithm is called an \emph{$(\epsilon,\delta)$-tester} for a property $P$, if given an input $I$, it accepts with probability at least $1-\delta$ if the input satisfies $P$, and it rejects with probability at least $1-\delta$ if the input $I$ is $\epsilon$-far from satisfying $P$.
Roughly speaking, we say that $I$ is \emph{$\epsilon$-far} from $P$ if we must modify at least an $\epsilon$-fraction of $I$ to make $I$ satisfy $P$.
When $\delta=1/3$, we simply call it an \emph{$\epsilon$-tester}.
A tester is called a \emph{one-sided error tester} if it always accepts when $I$ satisfies $P$.
In contrast, a standard tester is sometimes called a \emph{two-sided error tester}.
As one motivation of property testing is
to design algorithms that run in time sublinear in the input size, we assume query access to the input, and we measure the efficiency of a tester by its \emph{query complexity}.
We refer to~\cite{Goldreich:2011cg,Ron:2010ua,Rubinfeld:2011ik} for surveys on property testing.

In \emph{constraint satisfaction problems} (for short, \emph{$\CSP$s}),
one is
given a set of variables and a set of constraints imposed on the variables, and the task is to find an assignment of the variables that satisfies all of the given constraints.
By restricting the relations used to specify  constraints, it is known that certain restricted versions of the CSP coincide with many fundamental problems such as SAT, graph coloring, and
solvability of systems of linear equations.
To formally define these restricted versions of the CSP
(and hence, these problems),
we consider \emph{relational structures} $\bfA = (A; \Gamma)$,
where $A$ is a
finite set and $\Gamma$ consists of a finite set of finitary relations over $A$.  In this context, $\Gamma$
is sometimes referred to as a \emph{constraint language} over $A$ and $\bfA$ as a \emph{template}.
Then, we define $\CSP(\bfA)$ to be  those instances of the CSP whose constraint relations are taken from $\Gamma$.
In recent years, computational aspects of $\CSP(\bfA)$ have been
heavily studied, in the decision setting~\cite{IMMVW10-tractabilityfewsubpowers,Bulatov11-conservative-csp,Barto:2014ul,Barto:2014gn},
in counting complexity~\cite{Bulatov13-counting,DyerR13},
in computational learning theory~\cite{IMMVW10-tractabilityfewsubpowers,Chen:2015wh},
and in optimization and approximation~\cite{Raghavendra08,DBLP:conf/focs/ThapperZ12,ChanLRS13,DBLP:conf/stoc/ThapperZ13,DBLP:conf/icalp/ThapperZ15}.
See also the survey by Barto~\cite{barto-survey} for an overview of this line of research.

In this paper, we consider the problem family $\CSP(\bfA)$
from the perspective of property testing,
in particular, we consider the task of
testing assignments to CSPs.
Relative to a relational structure $\bfA$,
an input consists of a tuple $(\caI,\epsilon,f)$, where $\caI$ is an instance of $\CSP(\bfA)$ with weights on the variables, $\epsilon$ is an error parameter, and $f$ is an assignment to $\caI$.
In the studied model, the tester has
full access to $\caI$ and  query access to $f$, that is,
a variable $x$ can be queried to obtain the value of $f(x)$.
In this sense, assignment testing lies in the \emph{massively parameterized model}~\cite{Newman:2010da}.
We say that $f$ is \emph{$\epsilon$-far} from satisfying $\caI$ if one must modify at least an $\epsilon$-fraction of $f$ (with respect to the weights) to make $f$ a satisfying assignment of $\caI$, and we say that $f$ is \emph{$\epsilon$-close} otherwise.
It is always assumed that $\caI$ has a satisfying assignment as otherwise we can immediately reject the input (in this context,
one does not care about time complexity).
The objective of assignment testing of CSPs is to correctly decide whether $f$ is a satisfying assignment of $\caI$ or is $\epsilon$-far from being so with probability at least $2/3$.
When $f$ does not satisfy $\caI$ but is $\epsilon$-close to satisfying $\caI$, we can output anything.

In assignment testing, we say that the query complexity of a tester is constant/sublinear/linear if it is constant/sublinear/linear in the number of variables of an instance.
The main problem addressed in  this paper is to reveal the relationship between a relational structure $\bfA$ and the number of queries needed to test $\CSP(\bfA)$ and a related problem class $\eCSP(\bfA)$.

\subsection{Contributions}

  While previous works on testing assignments to the problems $\CSP(\bfA)$
  studied concrete templates $\bfA$ or restricted classes of structures, this article presents comprehensive classification theorems.

The first contribution of this paper is a dichotomy theorem that
\emph{completely characterizes} the constant-query testable CSPs.
Before describing our characterization, we introduce the
algebraic notion of a polymorphism which is key to the description
and obtention of our results.
Let $R$ be an $r$-ary relation on a set $A$.
A ($k$-ary) operation $f : A^k \to A$ is said to be a \emph{polymorphism} of $R$ (or $R$ is \emph{preserved} by $f$) if for any set of $k$ $r$-tuples $(a_1^1,\ldots,a_r^1), (a_1^2,\ldots,a_r^2), \ldots, (a_1^k,\ldots,a_r^k) \in R$, the tuple
  $(f(a_1^1,\ldots,a_1^k),\ldots,f(a_r^1,\ldots,a_r^k))$
also belongs to $R$.
An operation $f$ is a \emph{polymorphism} of a relational structure $\bfA$ if it is a polymorphism of each of its relation.
We define the \emph{algebra} of $\bfA$, denoted by $\Alg(\bfA)$, to be the pair  $(A; \Pol(\bfA))$, where $\Pol(\bfA)$ is the set of all polymorphisms of $\bfA$.

\begin{definition}
  Let $A$ be a nonempty set.
A \emph{majority operation} on $A$  is a ternary operation $m: A^3 \to A$ such that
 $m(b,a,a) = m(a,b,a) = m(a,a,b) = a$ for all $a$, $b \in A$.
A \emph{Maltsev operation} on $A$  is a ternary
operation $p: A^3 \to A$ such that
$p(b,a,a) = p(a,a,b) = b$ for all $a$, $b \in A$. For $k \geq 2$,
an operation $n: A^{(k+1)} \to A$
 is a \emph{$(k+1)$-ary near unanimity operation} on $A$ if for all $a$, $b \in A$,
  \[
  n(b,a,a,\ldots,a) = n(a,b,a,\ldots, a) = \cdots = n(a,a,\ldots, a,b) = a.
  \]
  (Note that a majority operation is a $3$-ary near-unanimity operation.)
\end{definition}
\begin{theorem}\label{the:intro-csp-main}
  Let $\bfA$ be a relational structure.
  The following dichotomy holds.
  \begin{itemize}
  \itemsep=0pt
  \item[(1)] If $\bfA$ has a majority polymorphism and a Maltsev  polymorphism, then $\CSP(\bfA)$ is constant-query testable (with one-sided error).
  \item[(2)] Else, testing $\CSP(\bfA)$ requires a super-constant number of queries.
  \end{itemize}
\end{theorem}

This theorem generalizes characterizations of constant-query testable List $H$-homomorphisms~\cite{Yoshida:2014zn} and Boolean CSPs~\cite{Bhattacharyya:2013fa} to general CSPs.
In Section~\ref{sec:arithmetic} we will describe the particularly nice structure of relations over templates that have majority and Maltsev polymorphisms and use this to prove the theorem.
For the moment, let us consider a number of example templates
to which the positive result of this theorem applies.
\begin{example}
The template $\bfA$ over the Boolean domain $\{0,1\}$ whose only relation is $\neq$
 has both majority and Maltsev polymorphisms.
 Note that $\CSP(\bfA)$ coincides with the graph 2-coloring problem.

More generally, the template $\bfA$ over a finite domain where each relation is a bijection on $A$ has both majority and Maltsev polymorphisms, and
instances of $\CSP(\bfA)$ for such templates $\bfA$ coincide with instances of the
problem which is the subject of the
\emph{unique games conjecture}~\cite{Khot:2002ju}.
\end{example}

\begin{example}
\label{ex:discriminator}
Another class of finite structures that have both majority and Maltsev polymorphisms are those that have a \emph{discriminator} operation as a polymorphism.  On a set $A$ the discriminator operation $d_A(x,y,z)$ is the operation such that if $x = y$ then $d(x,y,z) = z$ and if $x \ne y$, $d(x,y,z) = x$.  From this definition, it is immediate that $d_A$ is a Maltsev operation on $A$, and that $d(x, d(x,y,z),z)$ is a majority operation on $A$.  Any finite product of finite fields will have a discriminator term operation (\cite{Bu-Sa}) and so any finite relational structure whose relations are compatible with the operations of such a ring will have majority and Maltsev polymorphism.
\end{example}

 \begin{example}\label{RingExample}  
For $p$ a prime number, let $\bbF_p$ be the field of size $p$, and let $\bbR$ be the ring $\bbF_2 \times \bbF_3 \times \bbF_5$.
 Then as noted in Example~\ref{ex:discriminator}, $\bbR$ has a discriminator term operation.  Let $\bfR$ be the structure with domain $R$ and set of relations $\Gamma$  consisting of intersections of the following binary relations on $R$: For $p = 2$, 3, or 5,
\begin{itemize}
\item  $C_p = \{((a_2, a_3, a_5), (b_2, b_3, b_5)) \mid a_p = b_p\}$,
\item For $a \in \bbF_p$, $E_a = \{((a_2, a_3, a_5), (b_2, b_3, b_5)) \mid a_p = a\}$,
\item For $b \in \bbF_p$, $E_b = \{((a_2, a_3, a_5), (b_2, b_3, b_5)) \mid b_p = b\}$,
\end{itemize}
So relations in $\Gamma$ can express that pairs of elements in $R$ are congruent modulo 2, 3, or 5 in the corresponding coordinate and/or that a certain coordinate is equal to some fixed value.  These relations are invariant under the discriminator term operation of $\bbR$ and so according to Theorem~\ref{the:intro-csp-main}, $\CSP(\bfR)$ has constant query complexity.
\end{example}

Examples of structures that satisfy the first condition of Theorem~\ref{the:intro-csp-main} but that do not have a discriminator operation as a polymorphism can be derived from finite Heyting algebras.

\begin{example}
  Consider the five-element Heyting algebra $\mathbb{M}$ presented in~\cite[Figure 1]{IdziakIdziak}. (Heyting algebras are bounded distributive lattices that also have a binary ``implication'' operation; they serve as algebraic models of propositional intuitionistic logic.) This algebra has universe $M = \{0, a, b, e, 1\}$;
   the two
    equivalence relations $\alpha$ and $\beta$ that partition $M$ into blocks $\{\{0,a\}, \{b, e, 1\}\}$ and $\{\{0,b\}, \{a, e, 1\}\}$
    (respectively) are preserved by the operations of the algebra.
    Since $\mathbb{M}$ has majority and Maltsev term operations (the operations $(x\meet y)\join (x\meet z) \join (y\meet z)$ and $(x\rightarrow y)\rightarrow z) \meet(z\rightarrow y)\rightarrow x)$ respectively), then the structure $\mathbf{M} = (M; \alpha, \beta)$ has majority and Maltsev polymorphisms.  The only other non-trivial binary relation on $M$ that is definable by a primitive-positive formula over $\mathbf{M}$ is $\alpha\cap \beta$.
\end{example}

\begin{example}
 Bulatov and Marx provide yet another example
of a structure having both a majority and a Maltsev polymorphism,
 in~\cite[Example 1.1]{Bulatov-Marx}.
\end{example}

We next consider existentially quantified CSPs ($\eCSP$s for short).
The difference between CSPs and $\eCSP$s
is that, in an instance of $\eCSP$, existentially quantified variables may appear.
So, an instance  of $\eCSP$ may be defined as a primitive positive formula (pp-formula) over a relational structure.
Primitive positive formulas are known as \emph{conjunctive queries}
in the database theory literature; they are arguably the
most heavily studied class of database queries, and
the problem $\eCSP$ can be associated with the problem of \emph{conjunctive query evaluation}.

For a relational structure $\bfA = (A; \Gamma)$, we define $\eCSP(\bfA)$ to be the collection of instances of $\eCSP$ whose constraint relations  are taken from $\Gamma$.
Our second contribution is to provide a complete classification of
all structures $\bfA$
in terms of the number of queries needed to test assignments of instances of $\eCSP(\bfA)$ with one-sided error:
\begin{theorem}\label{the:intro-ecsp-main}
  Let $\bfA$ be a relational structure.
  Then, the following trichotomy holds.
  \begin{itemize}
  \itemsep=0pt
  \item[(1)] If $\bfA$ has a majority polymorphism and a Maltsev polymorphism, then $\eCSP(\bfA)$ is constant-query testable with one-sided error.
  \item[(2)] Else, if $\bfA$ has a $(k+1)$-ary near-unanimity polymorphism for some $k \geq 2$, and no Maltsev polymorphism then $\eCSP(\bfA)$ is not constant-query testable (even with two-sided error) but is sublinear-query testable with one-sided error.\footnote{
%
We remark that the combination of having a $(k+1)$-ary near unanimity polymorphism for some $k \ge 2$ and a Maltsev polymorphism is equivalent to having  majority and  Maltsev polymorphisms~\cite{Bu-Sa}.
  }
  \item[(3)] Else, testing $\eCSP(\bfA)$ with one-sided error requires a linear number of queries.
  \end{itemize}
\end{theorem}
Let us point out that the problem families $\CSP(\bfA)$ and
$\eCSP(\bfA)$ exhibit the same dichotomy for constant-query testability,
and in particular the positive result there is robust
with respect to the introduction of quantifiers.
An implication of Theorem~\ref{the:intro-ecsp-main} is this:
if the dichotomy for sublinear-query testability
for $\CSP(\bfA)$ was not the same as that for $\eCSP(\bfA)$,
then the positive result for that dichotomy would not
enjoy this robustness property, and hence such a positive result
would have to crucially exploit the absence of quantifiers.
Hence Theorem~\ref{the:intro-ecsp-main} reveals information
about the form of a potential trichotomy for $\CSP(\bfA)$.

A special feature of templates $\bfA$ that have a $(k+1)$-ary near-unanimity polymorphism is that any relation that is definable by a pp-formula over $\bfA$ can be decomposed into a number of $k$-ary relations that are also  pp-definable over $\bfA$.

\begin{example}
Consider
 the relational structure
$\bfA$
 over the Boolean domain $\{0,1\}$ whose only relation is $\le$.
This structure is readily verified to have a majority polymorphism
(note that over the Boolean domain, there is indeed a unique
majority operation), and does not have a Maltsev polymorphism:
for any Maltsev operation $p$, it holds that
applying $p$ to the tuples $(1,1),(0,1),(0,0)$, which are in the relation $\le$,
yields $(1,0)$, which are not in the relation $\le$.
Thus, Theorem~\ref{the:intro-ecsp-main} implies that
$\eCSP(\bfA)$ is not constant-query testable but is sublinear-query testable with one-sided error.
\end{example}

\begin{example}
We can generalize the previous example as follows.
Let $D$ be any finite set of size greater than or equal to $2$,
and consider the \emph{dual discriminator operation $\Delta$}
defined as follows: $\Delta(x,y,z)$ is equal to $x$ if $x=y$,
and is equal to $z$ otherwise.
Consider the relational structure $\bfA$ with universe $D$
and the following relations: each unary relation; each graph of a permutation (on $D$); and, each \emph{two-fan relation}.
Here, a \emph{two-fan relation} is a binary relation
$R \subseteq D \times D$ such that there exist
elements $a, b  \in D$ where
$R = (\{a\} \times \pi_2(R)) \cup (\pi_1(R) \times \{ b \})$.
It is straightforward to verify that
$\Delta$ is a majority polymorphism of $\bfA$.
On the other hand,
let $a, b \in D$ be arbitrary elements,
and consider the relation
$S = (\{ a \} \times D) \cup (D \times \{ b \})$.
The relation $S$ is a two-fan relation and so is a relation of $\bfA$,
but does not have a Maltsev polymorphism; we argue this as follows.
Let $a'$ be an element of $D$ distinct from $a$,
and let $b'$ be an element of $D$ distinct from $b$.
We have that the tuples $(a',b),(a,b),(a,b')$ are in $S$,
but if we apply any Maltsev polymorphism $p$ to them,
we obtain $(a',b')$ which is not in $S$.
The structure $\bfA$ thus does not have a Maltsev polymorphism;
we obtain by
Theorem~\ref{the:intro-ecsp-main}  that
$\eCSP(\bfA)$ is not constant-query testable but is sublinear-query testable with one-sided error.
\end{example}

\subsection{Proof outline}

We now describe outlines of our proofs of Theorems~\ref{the:intro-csp-main} and~\ref{the:intro-ecsp-main}.

\paragraph{$\bfA$ has majority and Maltsev polymorphisms $\Rightarrow$ $\eCSP(\bfA)$ is constant-query testable.}

We first look at~(1) of Theorem~\ref{the:intro-csp-main} and~\ref{the:intro-ecsp-main}.
As $\eCSP$s are a generalization of $\CSP$s, it suffices to consider $\eCSP$s.
Let $(\caI,\epsilon,f)$ be an input of the assignment testing of $\eCSP(\bfA)$.
First, we preprocess $\caI$ so that it becomes $2$-consistent (see Section~\ref{sec:pre} for the formal definition).
Using the 2-consistency of $\caI$ and the majority polymorphism of $\bfA$ we can assume that for each variable $x$ of $\caI$, the set of allowed values for $x$ forms a domain $A_x$ that is the universe of an algebra $\bbA_x$ that is a factor (i.e., a homomorphic image of a subalgebra) of  $\Alg(\bfA)$, the algebra of polymorphisms of $\bfA$. Also, we can assume that for each pair of variables $x$, $y$ of $\caI$  there is a unique binary constraint of $\caI$ with scope $(x,y)$ and constraint relation $R_{xy}$, with $R_{xy}$ the universe of some subalgebra of $\bbA_x \times \bbA_y$.  Furthermore  these are the only constraints  of $\caI$.

In order to test whether $f$ satisfies $\caI$, we use three types of reductions: a factoring reduction, a splitting reduction, and  an isomorphism reduction.
Each reduction produces an instance $\caI'$ and an assignment $f'$ such that $f'$ satisfies $\caI'$ if $f$ satisfies $\caI$, and $f'$ is $\Omega(\epsilon)$-far from satisfying $\caI'$ if $f$ is $\epsilon$-far from satisfying $\caI$.
For simplicity, we focus on how we create a new instance $\caI'$ here.

The objective of the factoring reduction is to factor, for each variable $x$ of $\caI$, the domain $A_x$ by any congruence $\theta$ of $\bbA_x$ (i.e., an equivalence relation on $A_x$ that is compatible with the operations of $\bbA_x$) for which none of the constraint relations of $\caI$ distinguish between $\theta$-related values of $A_x$.

%

After ensuring that all of the domains $A_x$ of $\caI$ cannot be factored, we then employ a  splitting reduction to ensure that for each variable $x$ of $\caI$ the algebra $\bbA_x$ is subdirectly irreducible, i.e., cannot be represented as a subdirect product of non-trivial algebras.
For any  variable $x$ for which $\bbA_x$ can be represented as a subdirect product of non-trivial algebras $\bbA_x^1$ and $\bbA_x^2$ we replace the variable $x$ by the new variables $x_1$ and $x_2$ and the domain $A_x$ by the domains $A_x^1$ and $A_x^2$.  For any other variable $y$ of $\caI$, we  ``split'' the constraint relation $R_{yx}$ (and its inverse $R_{xy}$) into two  relations $R_{yx_1}$ and $R_{yx_2}$ that are  together equivalent to the original one.  We then add these two new relations (and their inverses) to $\caI$, along with $A_x$, now regarded as a binary relation from the variable $x_1$ to $x_2$.

%

After performing the splitting reduction and the factoring reduction, we next define a binary relation~$\sim$ on the set of variables of $\caI$ such that $x \sim y$ if and only if the constraint relation $R_{xy}$ is the graph of an isomorphism from $\bbA_x$ to $\bbA_y$.  Using  2-consistency and the fact that the domains of $\caI$ are subdirectly irreducible and cannot be factored, it follows that, unless $\caI$ is trivial, the relation $\sim$ will be a non-trivial equivalence relation.  Within each $\sim$-class, the domains are isomorphic via the corresponding constraint relations of $\caI$, and this allows us to  produce an isomorphism-reduced instance $\caI'$ by restricting $\caI$ to a set of variables representing each of the $\sim$-classes.

After performing this isomorphism reduction, the resulting instance may have domains which can be further factored, allowing us to apply the factoring reduction to produce a smaller instance. We show that if we reach a point at which none of the three reductions can be applied, the instance must be trivial, either having just a single variable, or for which $|A_x| = 1$ for all variables $x$.  We also show that this point will be reached after applying the reductions at most $|A|$-times.

In Section~\ref{sec:arithmetic}, we will see how these reductions work on the template in Example~\ref{RingExample}.

%

\paragraph{$\CSP(\bfA)$ is constant-query testable $\Rightarrow$ $\bfA$ has majority and Maltsev polymorphisms.}
Now we look at~(2) of Theorem~\ref{the:intro-csp-main} and the hardness part of~(2) of Theorem~\ref{the:intro-ecsp-main}.
As $\eCSP$s are a generalization of $\CSP$s, it suffices to consider $\CSP$s.
We show that if $\bfA$ does not have these two types of polymorphisms, then we cannot test $\CSP(\bfA)$ with a constant number of queries.  We use that having these two types of polymorphisms is equivalent to $\bfA$ having a Maltsev polymorphism and that the variety of algebras generated by $\Alg(\bfA)$ is congruence meet semidistributive~\cite{Ho-McK}.  The paper~\cite{kkvw} provides a characterization of this condition in terms of the existence of two special polymorphisms of $\bfA$.
When the variety generated by $\Alg(\bfA)$ is not congruence meet semidistributive, then it can be easily shown from~\cite{Bhattacharyya:2013fa,Yoshida:2014zn} that testing $\CSP(\bfA)$ requires a linear number of queries.
When $\bfA$ does not have a Maltsev polymorphism, then we can reduce $\CSP(\bfA')$ to $\CSP(\bfA)$, where the structure $\bfA'$ has a binary non-rectangular relation.
Then, by replacing the $2$-SAT relations with this binary non-rectangular relation, we can reuse the argument for showing a super-constant lower bound for $2$-SAT in~\cite{Fischer:2002ev} to obtain a super-constant lower bound for $\CSP(\bfA)$.

\paragraph{$\bfA$ has a $(k+1)$-ary near-unanimity polymorphism,
for some $k \geq 2$ $\Rightarrow$ $\eCSP(\bfA)$ is sublinear-query testable.}

Now we consider the testability part of~(2) of Theorem~\ref{the:intro-ecsp-main}.
It is known that, if $\bfA$ has a $(k+1)$-ary near unanimity polymorphism, then $\CSP(\bfA)$ is sublinear-query testable with one-sided error in the unweighted case~\cite{Bhattacharyya:2013fa}.
We slightly modify their argument so that we can handle weights.

\paragraph{$\eCSP(\bfA)$ is sublinear-query testable with one-sided error $\Rightarrow$ $\bfA$ has a $(k+1)$-ary near unanimity polymorphism, for some $k \geq 2$.}
Finally, we consider~(3) of Theorem~\ref{the:intro-ecsp-main}.
We use that $\bfA$ has a $(k+1)$-ary near unanimity polymorphism for some $k \geq 2$ if and only if the variety of algebras generated by $\Alg(\bfA)$ is  congruence meet semidistributive and congruence modular~\cite{Ho-McK, Barto:2013uo}.
We already mentioned that if this variety is not congruence meet semidistributive then $\eCSP(\bfA)$ is not sublinear-query testable (even with two-sided error).
To complete the argument we show that if the variety is not congruence modular then,
by building on ideas developed in~\cite{Chen:2015wh},
we can reduce
the problem of testing assignments of a circuit in monotone $\mathrm{NC}^1$ to $\eCSP(\bfA)$.
Note that majority functions are in monotone $\mathrm{NC}^1$~\cite{Valiant:1984eu}, and we can easily show a linear lower bound for one-sided error testers that test assignments of majority functions.
Hence, we get a linear lower bound for $\eCSP(\bfA)$.

\subsection{Related work}

Assignment testing of CSPs was implicitly initiated by~\cite{Fischer:2002ev}. 
There, it was shown that 2-CSPs are testable with $O(\sqrt{n})$ queries and require $\Omega(\log n/\log \log n)$ queries for any fixed $\epsilon > 0$.
On the other hand, $3$-SAT~\cite{BenSasson:2005we}, $3$-LIN~\cite{BenSasson:2005we}, and Horn SAT~\cite{Bhattacharyya:2013fa} require $\Omega(n)$ queries to test.

The universal algebraic approach was first used in~\cite{Yoshida:2014zn} to study the assignment testing of the list $H$-homomorphism problem.
For graphs $G$, $H$, and list constraints $L_v \subseteq V(H)\;(v \in V(G))$,
we say that a mapping $f:V(G) \to V(H)$ is a \emph{list homomorphism} from $G$ to $H$ with respect to the list constraints $L_v\;(v \in V(G))$ if $f(v) \in L_v$ for any $v \in V(G)$ and $(f(u),f(v)) \in E(H)$ for any $(u,v) \in E(G)$.
Then, the corresponding assignment testing problem, parameterized by a graph $H$, is the following:
The input is a tuple $(G,\set{L_v}_{v \in V(H)},f,\epsilon)$,
where $G$ is a (weighted) graph, $L_v \subseteq V(H)\;(v \in V(G))$ are list constraints, $f:V(G)\to V(H)$ is a mapping given as a query access, and $\epsilon$ is an error parameter.
The goal is testing whether $f$ is a list $H$-homomorphism from $G$ or $\epsilon$-far from being so, where $\epsilon$-farness is defined analogously to testing assignments of CSPs.
It was shown in~\cite{Yoshida:2014zn} that the algebra (or the variety) associated with the list $H$-homomorphism characterizes the query complexity, and that list $H$-homomorphism is constant-query (resp., sublinear-query) testable  if and only if $H$ is a reflexive complete graph or an irreflexive complete bipartite graph (resp., a bi-arc graph).

Testing assignments of Boolean CSPs was studied in~\cite{Bhattacharyya:2013fa}, and in that paper relational structures were classified into three categories: (i) structures $\bfA$ for which $\CSP(\bfA)$ is constant-query testable, (ii) structures $\bfA$ for which $\CSP(\bfA)$ is not constant-query testable but sublinear-query testable, and (iii) structures $\bfA$ for which $\CSP(\bfA)$ is not sublinear-query testable.
They also relied on the fact that algebras (or varieties) can be used to characterize  query complexity.


\subsection{Open problems}

Theorem~\ref{the:intro-csp-main} characterizes relational structures $\bfA$ on general domains for which $\CSP(\bfA)$ is constant-query testable.
Obtaining a characterization for the sublinear-query testable case is a tantalizing open problem.
The main obstacle of this is that we obtained~(3) of Theorem~\ref{the:intro-ecsp-main} by reducing the problem of testing assignments of monotone circuits to $\eCSP$s.
If we do not allow existentially quantified variables, then the number of variables blows up polynomially in the reduction, and a linear lower bound for monotone circuits does not imply a linear lower bound for CSPs.

Theorem~\ref{the:intro-ecsp-main} provides a trichotomy for $\eCSP$s in terms of the number of queries needed to test with one-sided error.
Obtaining a similar trichotomy for two-sided error testers is also an interesting open problem.
Again the obstacle is that we reduce from the problem of testing assignments of monotone circuits.
It is not clear whether this problem is hard even for two-sided error testers.

\subsection{Organization}
Section~\ref{sec:pre} introduces the basic notions used throughout this paper.
We show the constant-query testability of $\eCSP$s with majority and Maltsev polymorphisms in Section~\ref{sec:arithmetic}.
Super-constant lower bounds of CSPs without majority or Maltsev polymorphisms is discussed in Section~\ref{sec:non-cp}.
We give a sublinear-query tester for $\eCSP$s having a $(k+1)$-ary near unanimity polymorphism, for some $k \ge 2$, in Section~\ref{sec:cd}.
In Section~\ref{sec:non-cd}, we show that, when there is no $(k+1)$-ary near unanimity polymorphism for any $k \ge 2$,  testing $\eCSP$s with one-sided error requires a linear number of queries.


\section{Preliminaries}\label{sec:pre}

For an integer $k$, let $[k]$ denote the set $\set{1,\ldots,k}$.


\paragraph{Constraint satisfaction problems}
For an integer $k \geq 1$, a $k$-ary relation on a domain $A$ is a subset of $A^k$.
A \emph{constraint language} on a domain $A$ is a finite set of relations on $A$.
A \textit{(finite) relational structure}, or simply a \emph{structure} $\bfA = (A; \Gamma)$ consists of a non-empty set $A$ and a constraint language $\Gamma$ on $A$.

For a structure $\bfA = (A; \Gamma)$, we define the problem $\CSP(\bfA)$ as follows.
An instance $\caI = (V,A,\caC,\bw)$ consists of a set of variables $V$, a set of constraints $\caC$, and a weight function $\bw$ with $\sum_{x \in V}\bw(x) = 1$.
Here, each constraint $C \in \caC$ is of the form $\langle (x_1,\ldots,x_k),R\rangle$, where $x_1,\ldots,x_k \in V$ are variables, $R$ is a relation in $\Gamma$ and $k$ is the arity of $R$.
An \emph{assignment} for $\caI$ is a mapping $f:V \to A$, and we say that $f$ is a \emph{satisfying assignment} if $f$ satisfies all the constraints, that is, $(f(x_1),\ldots,f(x_k)) \in R$ for every constraint $\langle (x_1,\ldots,x_k),R\rangle \in \caC$.

For a structure $\bfA = (A; \Gamma)$, we define the problem $\eCSP(\bfA)$ as follows.
An instance $\caI = (V,V^\exists, \caC,\bw)$ consists of a set of free variables $V$, a set of existentially quantified variables $V^\exists$, a set of constraints $\caC$, and a weight function $\bw$ with $\sum_{x \in V}\bw(x) = 1$.
Constraints are imposed on $V \cup V^\exists$.
An \emph{assignment} for $\caI$ is a mapping $f:V \to A$, and we say that $f$ is a \emph{satisfying assignment} if there exists an extension $f':V \cup V^\exists \to A$ of $f$ such that $(f'(x_1),\ldots,f'(x_k)) \in R$ for every constraint $\langle (x_1,\ldots,x_k),R\rangle \in \caC$.




\paragraph{Algebras and Varieties:}
Let $\bbA = (A; F)$ be an algebra.
A set $B \subseteq A$ is a \textit{subuniverse} of $\bbA$ if for every operation $f \in F$ restricted to $B$ has
image contained in $B$.
For a nonempty subuniverse $B$ of an algebra $\bbA$, $f|_{B}$ is the
restriction of $f$ to $B$.
The algebra $\bbB = (B,F|_{B})$, where $F|_{B} = \{f|_{B} \mid f \in F\}$
is a \textit{subalgebra} of $\bbA$.
Algebras $\bbA,\bbB$ are of the \textit{same type} if they have the same number of operations and corresponding operations have the same arities.
Given algebras $\bbA,\bbB$ of the same type, the \textit{product}
$\bbA\times \bbB$ is the algebra with the same type as $\bbA$ and $\bbB$
with universe $A \times B$ and operations computed
coordinate-wise.  A subalgebra $\bbC$ of $\bbA\times \bbB$ is a \emph{subdirect product} of $\bbA$ and $\bbB$ if the projections of $C$ to $A$ and $C$ to $B$ are both onto.
An equivalence relation $\theta$ on $A$ is called a \textit{congruence} of an algebra $\bbA$ if $\theta$ is a subalgebra of $\bbA \times \bbA$.  The collection of congruences of an algebra naturally forms a lattice under the inclusion ordering, and this lattice is called the \emph{congruence lattice} of the algebra.
Given a congruence $\theta$ on $A$, we can form the \textit{homomorphic image} $\bbA /_{\theta}$, whose elements are the equivalence classes
of $\bbA$ and the operations are defined so that the natural
 mapping from $\bbA$ to $\bbA/_{\theta}$
 is a homomorphism.  An operation $f(x_1, \dots x_n)$ on a set $A$ is \emph{idempotent} if $f(a, a, \dots, a) = a$ for all $a \in A$, an algebra $\bbA$ is \emph{idempotent} if each of its operations is, and a class of algebras is idempotent if each of its members is.  We note that if $\bbA$ is \emph{idempotent}, then for any congruence $\theta$ of $\bbA$, the $\theta$-classes are all subuniverses of $\bbA$.


A \textit{variety} is a class of algebras of the same type closed under the formation of homomorphic images, subalgebras,
and  products.
For any algebra $\bbA$,
there is a smallest variety containing $\bbA$, denoted by $\caV(\bbA)$ and called the \textit{variety generated} by $\bbA$.
It is well known that any variety is generated by an algebra and that any member of $\caV(\bbA)$ is a homomorphic image of a subalgebra of a power of $\bbA$.

Many important properties of the algebras in a variety can be correlated with properties of the congruence lattices of it member algebras.  In this work we consider several congruence lattice conditions for varieties, including congruence modularity, congruence distributivity, congruence meet semidistributivity, and congruence permutability.  Details of these conditions can be found in~\cite{Ho-McK} and  more details on the basics of algebras and varieties can be found in \cite{Bu-Sa}.

\subsection{Assignment problems}

An \emph{assignment problem} consists of a set of \emph{instances},
where each instance $\caI$ has associated with it
a set of variables $V$, a domain $A_v$ for each variable $v \in V$,
and a weight function
$\bw: V \to [0,1]$ with $\sum_{v \in V} \bw(v) = 1$.
An assignment of $\caI$
is a mapping $f$ defined on $V$
with $f(x) \in A_x$ for each variable $x \in V$.
Each instance $\caI$ of an assignment problem
has associated with it a notion of a \emph{satisfying assignment}.
For two assignments $f$ and $g$ for $\caI$, we define their distance as $\dist(f,g) := \sum_{x \in V:f(x) \neq g(x)}\bw(x)$.
We define $\dist_\caI(f) = \min_g\dist(f,g)$, where $g$ is over all satisfying assignments of $\caI$.
Then, for $\epsilon \in [0,1]$, we say that an assignment $f$ for $\caI$ is \emph{$\epsilon$-far from satisfying $\caI$} if $\dist_\caI(f) > \epsilon$.
In the \emph{assignment testing problem} corresponding to an assignment problem, we are given an instance $\caI$ 
of the assignment problem and a query access to an assignment $f$ for $\caI$, that is, we can obtain the value of $f(x)$ by querying $x \in V$.
Then, we say that an algorithm is a \emph{tester} for the assignment problem if it accepts with probability at least $2/3$ when $f$ is a satisfying assignment of $\caI$, and rejects with probability at least $2/3$ when $f$ is $\epsilon$-far from satisfying $\caI$.
The \emph{query complexity} of a tester is the number of queries to $f$.

We can naturally view $\CSP(\mathbf{A})$ and $\eCSP(\mathbf{A})$ as assignment problems:
for each instance on a set of (free) variables $V$,
the associated assignments are the mappings from $V$ to $A$,
and the notion of satisfying assignments is as described above.
Note that an input to the assignment testing problem corresponding to $\CSP(\bfA)$ or to $\eCSP(\bfA)$ is a tuple $(\caI,\epsilon,f)$, where $\caI$ is an instance of $\CSP(\bfA)$ or $\eCSP(\bfA)$, respectively, $\epsilon$ is an error parameter, and $f$ is an assignment to $\caI$.
In order to distinguish $\caI$ from the tuple $(\caI,\epsilon,f)$, we always call the former \emph{instance} and the latter \emph{input}.

%



\subsubsection{Gap-preserving local reductions}

We will frequently use the following reduction when constructing algorithms as well as showing lower bounds.
\begin{definition}[Gap-preserving local reduction]\label{def:gap-preserving-local-reduction}
  Given assignment problems $\caP$ and $\caP'$,
  there is
  a {\em (randomized) gap-preserving local reduction from $\caP$ to $\caP'$}
  if
  there exist a function $t(n)$ and constants $c_1,c_2$ satisfying the following:
  given a $\caP$-instance
  $\caI$ of with variable set $V$ and an assignment $f$ for $\caI$,
there exist a
$\caP'$-instance $\caI'$ with variable set $V'$
and an assignment $f'$ for $\caI'$
such that the following hold:
  \begin{enumerate}
  \item $|V'| \leq t(|V|)$.
  \item If $f$ is a satisfying assignment of $\caI$,
      then $f'$ is a satisfying assignment of $\caI'$.
  \item For any $\epsilon \in (0,1)$, if $\dist_{\caI}(f) \geq \epsilon$,
  then $\Pr[\dist_{\caI'}(f') \geq c_1\epsilon] \geq 9/10$ holds, where the probability is over internal randomness.
  \item Any query to $f'$ can be answered by making at most $c_2$ queries to $f$.
  \end{enumerate}
\end{definition}

A \emph{linear reduction} is defined to be a gap-preserving
local reduction for which the function $t(n) = O(n)$, $c_1 = O(1)$, and $c_2 = O(1)$.

\begin{lemma}[\cite{Yoshida:2014zn}]\label{lmm:gap-preserving-local-reduction}
  Let $\caP$ and $\caP'$ be assignment problems.
  Suppose that there exists an $\epsilon$-tester for $\caP'$
  with query complexity $q(n,\epsilon)$ for any $\epsilon \in (0,1)$,
  where $n$ is the number of variables in the given instance of $\caP'$, and that there exists a gap-preserving local reduction from $\caP$ to $\caP'$ with a function $t$ and $c_1=c_2=O(1)$.
  Then, there exists an $\epsilon$-tester for $\caP$ with query complexity $O(q(t(n), O(\epsilon)))$ for any $\epsilon > 0$, where $n$ is the number of variables in the given instance of $\caP$.
  In particular, linear reductions preserve constant-query and sublinear-query testability.
\end{lemma}

As another application of gap-preserving local reductions, the following fact is known.
\begin{lemma}[Lemma~6.4 and~6.5 of~\cite{Yoshida:2014zn}]\label{lem:reduction-for-csps}
  Let $\bfA, \bfA'$ be relational structures.
  If the relations of $\bfA$ are preserved by the operations of some finite algebra in $\caV(\Alg(\bfA'))$,
  then $\CSP(\bfA)$ is constant-query testable if $\CSP(\bfA')$ is constant-query testable.
\end{lemma}

In the proof of Lemma~\ref{lem:reduction-for-csps}, the only obstacle that prevents linear reductions is that the number of variables blows up by introducing new variables for each constraint.
However, we can get rid of this obstacle by replacing them with existentially quantified variables and we get the following.
\begin{lemma}\label{lem:reduction-for-ecsps}
  Let $\bfA, \bfA'$ be relational structures.
  If the relations of $\bfA$ are preserved by the operations of some finite algebra in $\caV(\Alg(\bfA'))$,
  then there exists a linear reduction from $\eCSP(\bfA)$ to $\eCSP(\bfA')$.
  In particular, $\eCSP(\bfA)$ is constant-query (resp., sublinear-query) testable if $\eCSP(\bfA')$ is constant-query (resp., sublinear-query) testable.
\end{lemma}

\subsubsection{$(k+1)$-ary near unanimity polymorphisms}

Let $\caI = (V,V^\exists,\caC,\bw)$ be an instance of $\eCSP(\bfA)$.
A \emph{partial solution} of $\caI$ on a set of variables $W \subseteq V \cup V^\exists$ is a mapping $\psi:W \to A$ that satisfies every constraint $\langle W \cap s, \pr_{W \cap s} R\rangle$ where $\langle s,R\rangle \in \caC$ and $\pr_{W \cap S}R$ is the projection of $R$ to $W \cap S$.
Here $W \cap s$ denotes the subtuple of $s$ consisting of those entries of $s$ that belong to $W$, and we consider the coordinate positions of $R$ indexed by variables from $s$.
Instance $\caI$ is said to be \emph{$k$-consistent} if for any $k$-element set $W \subseteq V \cup V^\exists$ and any $v \in (V \cup V^\exists) \setminus W$ any partial solution on $W$ can be extended to a partial solution on $W \cup \set{v}$.
It is well known that, for any constant $k$, any $\CSP$ instance can be transformed to a $k$-consistent instance in polynomial time without changing the set of satisfying assignments.
See~\cite{Jeavons:1998gb} for more details.

Let $\bfA = (A;\Gamma)$ be a relational structure with a $(k+1)$-ary near unanimity polymorphism,
and let $\caI = (V,V^{\exists},\caC,\bw)$ be a $k$-consistent instance of $\eCSP(\bfA)$.
Because of the existence of a $(k+1)$-ary near unanimity polymorphism, we can assume that every constraint is $k$-ary.
Hence, we can write $\caI = (V,V^{\exists},\set{R_S}_{(V \cup V^\exists)^k},\bw)$.
Further, we can say that, for any set $S \subseteq V \cup V^\exists$ of size $k$,
any partial assignment $f:S \to A$ with $f|_S \in R_S$ can be extended to a satisfying assignment for the whole instance~\cite{Feder:1998iu}.
This property is called the \emph{$k$-Helly property}.
We call a subset of variables $S \subseteq V$ of size $k$ \emph{violated} with respect to an assignment $f:V \to A$ if $f|_S \not \in R_S$.

As an application of gap-preserving local reductions, we observe that if a relational structure $\bfA$ has a $(k+1)$-ary near unanimity  polymorphism for some $k \ge 2$, then testing $\eCSP(\bfA)$ can be reduced to $\CSP(\bfA)$.
\begin{lemma}\label{lem:reduction-from-ecsp-to-csp}
  Let $\bfA$ be a relational structure with a $(k+1)$-ary near unanimity polymorphism for some $k \geq 2$.
  Then, there is a linear reduction from $\eCSP(\bfA)$ to $\CSP(\bfA)$.
\end{lemma}
\begin{proof}
  Let $\caI = (V,V^\exists,\set{R_S}_{S \in (V \cup V^\exists)^k},\bw)$ be a $k$-consistent instance of $\eCSP(\bfA)$.
  Then, we consider the instance $\caI' = (V,\set{R_S}_{S \subseteq V^k},\bw)$ and the assignment $f' = f$.

  If $f$ satisfies $\caI$, then $f'$ also satisfies $\caI'$ because the constraints of $\caI'$ are also constraints of $\caI$.

  Suppose that $f'$ is $\epsilon$-close to satisfying $\caI'$ and let $g'$ be a satisfying assignment of $\caI'$ with $\dist_{\caI'}(f',g') \leq \epsilon$.
  Then, we define $g = g'$.
  Note that $g$ satisfies $\caI$ because there is no violated constraint caused by $g$, and from the $k$-Helly property, we can always extend it to a satisfying assignment for the whole instance.
  Hence, $f$ is $\epsilon$-close to $g$.

  To summarize, this reduction is a gap-preserving local reduction with $t(n) = n$, $c_1 = 1$, and $c_2 = 1$.
\end{proof}


\section{Constant-Query Testability}\label{sec:arithmetic}

In this section, assume that $\bfA = (A;\Gamma)$ is a structure that has a majority polymorphism $m(x,y,z)$ and a Maltsev polymorphism $p(x,y,z)$. It is known,~\cite{Bu-Sa}, that this is equivalent to the variety $\caA$ generated by the algebra $\Alg(\bfA)$ being \emph{congruence distributive} and \emph{congruence permutable} and also to $\bfA$ having a $(k+1)$-ary near unanimity polymorphism for some $k \ge 2$ and a Maltsev polymorphism.  This means that for each algebra $\bbB \in \caA$, the lattice of congruences of $\bbB$ satisfies the distributive law and that for each pair of congruences $\alpha$ and $\beta$ of $\bbB$, the relations $\alpha \circ \beta$ and $\beta \circ \alpha$ are equal.  Such varieties are also said to be \emph{arithmetic}.

An important feature of $\caA$ (and in fact of any congruence distributive variety generated by a finite algebra) is that every subdirectly irreducible member of $\caA$ has size bounded by $|A|$~(\cite{Bu-Sa}). 
We will make use of the fact that an algebra is \emph{subdirectly irreducible} if and only if the intersection of all of its non-trivial congruences is non-trivial.  This is equivalent to the algebra having a smallest non-trivial congruence.
In this section, we will show that $\eCSP(\bfA)$ is constant-query testable.  Some of the ideas found in this section were inspired by the paper~\cite{Bulatov-Marx}.

We first note that, since $\bfA$ has a majority operation, that is, a $3$-ary near unanimity operation, as a polymorphism, it suffices to consider $\CSP(\bfA)$ by Lemma~\ref{lem:reduction-from-ecsp-to-csp}.

For our analysis, it is useful to introduce $\CSP(\caV)$ for a variety $\caV$.
An instance of $\CSP(\caV)$ is of the form $(V, \set{A_x}_{x \in V}, \caC,\bw)$.
Each $A_x$ is the domain of an algebra, denoted by $\bbA_x$, in $\caV$, and each constraint in $\caC$ is of the form $\langle (x_1,\ldots,x_k), R\rangle$, where $R$ is the domain of a subalgebra $\bbR$ of $\bbA_{x_1} \times \cdots \times \bbA_{x_k}$.
In particular, $R$ is also the domain of an algebra in $\caV$.
The definitions of $2$-consistency and an assignment testing problem naturally carry over to instances of $\CSP(\caV)$.


Let $\caI = (V, \set{A_x}_{x \in V}, \caC,\bw)$ be an instance of $\CSP(\caA)$.
Since $\caA$ is arithmetic, we can assume that each constraint in $\caC$ is binary~\cite{Baker:1975wn}.
Hence, we also write
\[
\caI = (V,\set{A_x}_{x \in V}, \set{R_{xy}}_{(x,y) \in V^2}, \bw)
\]
 or simply $\caI = (V, \set{A_x}, \set{R_{xy}}, \bw)$.
Moreover, we can assume that $\caI$ is $2$-consistent because the set of satisfying assignments does not change after making $\caI$ $2$-consistent.
For $x \in V$, $R_{xx}$ is the equality relation $0_{A_x}$ on the set $A_x$, and for distinct variables $x \neq y \in V$, $R_{xy}$ denotes the (unique) binary constraint relation from $A_x$ to $A_y$.
We always have $R_{yx} = R_{xy}^{-1} = \set{(b,a) \mid (a,b) \in R_{xy}}$ for any $x,y \in V$.
We note that by 2-consistency, it follows that for distinct variables $x$ and $y$, the relation $R_{x,y}$ is subdirect in $A_x \times A_y$.
Throughout the remainder of this section, we will assume that any instance of $\CSP(\caA)$ considered will be 2-consistent and has only binary constraints.

Since $\caA$ is assumed to be congruence permutable ( then for any $x \ne y \in V$, the binary relation $R_{xy}$ is \emph{rectangular}, that is, $(a,c),(a,d),(b,d) \in R_{xy}$ implies $(b,c) \in R_{xy}$.
As noted in Lemma 2.10 of~\cite{Bulatov-Marx}, this is equivalent to $R_{xy}$ being a \emph{thick mapping}.  This means that there are congruences $\theta_{xy}$ of $\bbA_x$ and $\theta_{yx}$ of $\bbA_y$ such that modulo the congruence $\theta_{xy}\times \theta_{yx}$ on $\bbR_{xy}$, the relation $R_{xy}$ is the graph of an isomorphism $\phi_{xy}$ from $\bbA_x/\theta_{xy}$ to $\bbA_y/\theta_{yx}$ and such that for all $a \in A_x$ and $b \in A_y$, $(a,b) \in R_{xy}$ if and only if $\phi_{xy}(a/\theta_{xy}) = b/\theta_{yx}$.  In this situation, we say that $R_{xy}$ is a thick mapping with respect to $\theta_{xy}$, $\theta_{yx}$ and $\phi_{xy}$.  For future reference, we note that if for some variables $x \ne y$, the congruence $\theta_{xy} = 0_{A_x}$ then the relation $R_{yx}$ is the graph of a surjective homomorphism from $\bbA_{y}$ to $\bbA_x$.

\subsection{A factoring reduction}

Let $\caI = (V,\set{A_x}, \set{R_{xy}}, \bw)$ be an instance of $\CSP(\caA)$ and for each $x \in V$
let $\mu_x = \bigwedge_{y\ne x} \theta_{xy}$, a congruence of $\bbA_x$.
We say that $A_x$ is \emph{prime} if $\mu_x$ is  the equality congruence $0_{A_x}$ and \emph{factorable} otherwise.
Roughly speaking, if $A_x$ is not prime, then we can factor $A_x$ by $\mu_x$ without changing the problem, because no constraint of $\caI$ distinguishes values within any $\mu_x$-class.
Formally, we define the factoring reduction as in Algorithm~\ref{alg:factoring-reduction}.
\begin{algorithm}
  \caption{}\label{alg:factoring-reduction}
  \begin{algorithmic}[1]
  \Procedure{Factor}{$\caI = (V, \set{A_x}, \set{R_{xy}}, \bw), \epsilon, f$}
    \For{$x \in V$}
    \State{$A_x \leftarrow A_x / \mu_x$.}
    \State{$f(x) \leftarrow f(x)/ \mu_x$.}
    \EndFor
    \For{$(x,y) \in V \times V$}
    \State{$R_{xy} \leftarrow \set{(a / \mu_x,b/\mu_y) \mid (a,b) \in R_{xy}}$.}
    \EndFor
    \State{\Return $(\caI,\epsilon,f)$.}
  \EndProcedure
  \end{algorithmic}
\end{algorithm}

Let $(\caI,\epsilon,f)$ be an input of $\CSP(\caA)$ and let $(\caI',\epsilon',f') = \Call{Factor}{\caI,\epsilon,f}$.
It is clear that since the instance $\caI$ of $\CSP(\caA)$ is assumed to be 2-consistent then the instance $\caI'$ will also be 2-consistent.  Furthermore, the sizes of the domains of $\caI'$ are no larger than the sizes of the domains of $\caI$.
Now we show that the factoring reduction is a linear reduction.
\begin{lemma}
  Let $(\caI,\epsilon,f)$ be an input of $\CSP(\caA)$ and let $(\caI',\epsilon',f') = \Call{Factor}{\caI,\epsilon,f}$.
  If $(\caI',\epsilon',f')$ is testable with $q(\epsilon')$ queries, then $(\caI,\epsilon,f)$ is testable with $q(O(\epsilon))$ queries.
\end{lemma}
\begin{proof}
  We show that the factoring reduction is a linear reduction.
  Let $\caI = (V, \set{A_x}, \set{R_{xy}}, \bw)$ and $\caI' = (V',\set{A'_x}_{x \in V}, \set{R'_{xy}}, \bw')$ be the original instance and the reduced instance, respectively.

  Note that $|V'| = |V|$ and we can determine the value of $f'(x)$ by querying $f(x)$.

  If $f$ satisfies $\caI$, then $f'$ also satisfies $\caI'$.
  Suppose that $f'$ is $\epsilon$-close to satisfying $\caI'$ and let $g'$ be a satisfying assignment of $\caI'$ with $\dist_{\caI'}(f',g') \leq \epsilon$.
  Then, we define  $g$ to be any assignment for $\caI$ such that for $x \in V$, $g(x)$ is taken to be an arbitrary element in the $\mu_x$-class  $g'(x)$.
  Then, $g$  satisfies $\caI$ and $\dist_{\caI'}(f,g) = \dist_{\caI}(f',g') \leq \epsilon$.

  To summarize, the factoring reduction is a gap-preserving local reduction with $t(n) = n$, $c_1 = 1$, and $c_2 = 1$.
\end{proof}

%

\subsection{Reduction to instances with subdirectly irreducible domains}

In this section, we provide a reduction that produces instances whose domains are all subdirectly irreducible.
Suppose that $\bbA$ is a subdirect product of two algebras $\bbA_1$, $\bbA_2$ from $\caA$ and that $\bbR$ is a subdirect product of $\bbA$ and $\bbB$ for some $\bbB \in \caA$.  We can project the relation $R$ onto the factors of $\bbA$ to obtain two new binary relations from $A_1$ to $B$ and from $A_2$ to $B$, respectively:
\begin{align*}
R_1 &= \{(a_1,b)\mid \text{there is some $(a_1, c_2) \in A$ with $((a_1, c_2), b) \in R$}\},\\
R_2 &= \{(a_2,b)\mid \text{there is some $(c_1, a_2) \in A$ with $((c_1, a_2), b) \in R$}\}.
\end{align*}
The following shows that the relation $R$ can be recovered from the relations $R_1$, $R_2$, and $A$ (considered as a relation from $A_1$ to $A_2$).

\begin{lemma}\label{lem:subdirectly-irreducible}
  For all $a_1 \in A_1$, $a_2 \in A_2$, and $b \in B$, the following are equivalent:
  \begin{itemize}
    \item $((a_1, a_2), b) \in R$
    \item $(a_1,b) \in R_1$, $(a_2,b) \in R_2$ and $(a_1, a_2) \in A$.
  \end{itemize}
\end{lemma}
\begin{proof}
  One direction of this claim follows by construction.
  For the other, suppose that $(a_1,b) \in R_1$, $(a_2,b) \in R_2$ and $(a_1, a_2) \in A$.  Then there are elements $c_i \in A_i$, for $i = 1$, 2, with $(a_1, c_2)$, $(c_1, a_2) \in A$, $((a_1, c_2), b)$, $((c_1, a_2), b) \in R$.
  Since $R$ is subdirect in $A\times B$ and $(a_1, a_2) \in A$ then there is some $d \in B$ with $((a_1, a_2), d) \in R$.  Applying the majority term of $\caA$ coordinate-wise to the tuples $((a_1, c_2), b)$, $((c_1, a_2), b)$, and $((a_1, a_2), d)$  from $R$ we produce the tuple $((a_1, a_2), b) \in R$, as required.
%
\end{proof}

Lemma~\ref{lem:subdirectly-irreducible} allows us to split a domain of an instance of $\CSP(\caA)$ into subdirectly irreducible domains.
Formally, we define the splitting reduction as in Algorithm~\ref{alg:splitting-reduction}.
\begin{algorithm}
  \caption{}\label{alg:splitting-reduction}
  \begin{algorithmic}[1]
  \Procedure{Split}{$\caI = (V, \set{A_x}, \set{R_{xy}}, \bw), \epsilon, f$}
  \While{there exists $x \in V$ such that $\bbA_x$ is not subdirectly irreducible or trivial}
   \State{Replace $\bbA_x$ in $\caI$ with an isomorphic non-trivial subdirect product of $ \bbA_x^1 \times \bbA_x^2$ for some quotients $ \bbA_x^1$,  $ \bbA_x^2$ of $\bbA_x$ such that $\bbA_{x}^1$ is subdirectly irreducible.}
  \State{$V \leftarrow (V \setminus \set{x} )\cup \set{x_1,x_2}$, where $x_1$ and $x_2$ are newly introduced variables.}
  \State{Remove the domain $A_x$ and add the domains $A_{x}^1$ and $A_{x}^2$ over the variables $x_1$ and $x_2$ respectively.}
  \State{$\caC \leftarrow \caC \setminus \set{\langle (x,x), R_{xx}\rangle, \langle (x,y),R_{xy}\rangle, \langle (y,x), R_{yx}\rangle}_{y \in V \setminus \set{x}}$.}
  \State{$\caC \leftarrow \caC\cup \set{\langle (x_1, x_1), 0_{A_{x_1}}\rangle, \langle (x_2, x_2), 0_{A_{x_2}}\rangle, \langle (x_1,x_2),A_x\rangle, \langle (x_2,x_1),A_x^{-1}\rangle}$.}
   \State{$\caC \leftarrow \caC \cup \set{\langle (x_1,y),(R_{xy})_1 \rangle,\langle (x_2,y),(R_{xy})_2 \rangle,\langle (y,x_1),(R_{xy})_1^{-1} \rangle,\langle (y,x_2),(R_{xy})_2^{-1} \rangle}_{y \in V \setminus \set{x}}$.}
  \State{Remove $x$ from the domain of $\bw$ and add $x_1$ and $x_2$.}
  \State{Set $\bw(x_1) = \bw(x)/2$ and $\bw(x_2) = \bw(x)/2$.}
  \State{Remove $x$ from the domain of $f$ and add $x_1$ and $x_2$.}
  \State{Set $f(x_1) \in A_x^1$ and $f(x_2) \in A_x^2$ so that $(f(x_1),f(x_2)) = f(x)$.}

  \EndWhile
  \State{\Return $(\caI,\epsilon/2^{|A|},f)$.}
  \EndProcedure
  \end{algorithmic}
\end{algorithm}

Let $(\caI,\epsilon,f)$ be an input of $\CSP(\caA)$ and let $(\caI',\epsilon',f') = \Call{Split}{\caI,\epsilon,f}$.
It is clear that, since $\caI$ is assumed to be a 2-consistent instance of $\CSP(\caA)$ then the splitting reduction constructs another 2-consistent instance $\caI'$ of $\CSP(\caA)$ whose domains are all subdirectly irreducible and so have size bounded by $|A|$ (and are no bigger than the domains of $\caI$).
The next lemma shows that if a domain of an instance $\caI$ is prime, then after splitting it, the resulting subdirect factors will also be prime.
\begin{lemma}
  Let $\caI'$ be the instance of $\CSP(\caA)$ obtained by splitting a domain $\bbA_x$ of another instance $\caI$ into two subdirect factors $\bbA_{x_1}$ and $\bbA_{x_2}$ as in the $\Call{Split}{}$ procedure.  If the domain $\bbA_x$ is prime in $\caI$ then the domains $\bbA_{x_1}$ and $\bbA_{x_2}$ are prime in $\caI'$.
\end{lemma}

\begin{proof}
  Let $\caI = (V, \set{A_x}, \set{R_{xy}}, \bw)$ be given and suppose that the domain $\bbA_x$ is a subdirect product of the algebras $\bbA_{x_1}$ and $\bbA_{x_2}$. To produce $\caI'$ from $\caI$ by splitting $\bbA_x$, we replace the variable $x$ and the domain $A_x$ with the variables $x_1$ and $x_2$ and the corresponding domains $A_{x_1}$ and $A_{x_2}$.  For each $y \in V$ with $x \ne y$, we replace the constraint $\langle (x,y), R_{xy}\rangle$ with the constraints $\langle (x_1,y), (R_{xy})_1\rangle$ and $\langle (x_2,y), (R_{xy})_2\rangle$ and add the constraint $\langle (x_1, x_2), A_x\rangle$.

  If the domain $\bbA_x$ is prime in $\caI$ then there is $k \ge 1$ and  variables $y_i \in V\setminus \{x\}$, for $1 \le i \le k$, such that $\bigwedge _{1 \le i \le k} \theta_{xy_i} = 0_{A_x}$.  To show that $\bbA_{x_1}$ is prime in $\caI'$ it will suffice to show that
  \[
  \left(\bigwedge_{1 \le i \le k}\theta_{x_1y_i} \right) \wedge \theta_{x_1x_2} = 0_{A_{x_1}}.
  \]
  To establish this, suppose that $(a_1, a_1')$ belongs to the left hand side of this equality.  We will show that $a_1 = a_1'$.  We have that $(a_1, a_1') \in \theta_{x_1y_i}$ for $1 \le i \le k$ and $(a_1, a_1') \in \theta_{x_1x_2}$.  From the latter membership it follows that there is some $c \in A_{x_2}$ such that $(a_1,c)$, $(a_1',c) \in A_x$.
  From $(a_1, a_1') \in \theta_{x_1y_i}$ it follows that there is some $u \in A_{y_i}$ with $(a_1,u)$, $(a_1',u) \in (R_{xy_i})_1$. We can conclude that there are $d$, $d' \in A_{y_i}$ with $((a_1, d), u)$, $((a_1', d'), u) \in R_{xy_i}$.  We then have that $((a_1, d), (a_1',d')) \in \theta_{xy_i}$.  We can now apply the majority term of $\caA$ coordinate-wise to the following three pairs of members of $\theta_{xy_i}$ to establish that $((a_1, c), (a_1',c)) \in \theta_{xy_i}$:
  $((a_1, d), (a_1',d'))$, $((a_1,c), (a_1,c))$, and $((a_1',c), (a_1',c))$.  We've shown that $(a_1,c)$ and $(a_1',c)$ are $\theta_{xy_i}$-related for all $i \le k$ and so we have that $(a_1,c) = (a_1',c)$, which implies that $a_1 = a_1'$, as required. Thus $\bbA_{x_1}$ is prime in $\caI'$ and by symmetry, $\bbA_{x_2}$ is also prime.
  \end{proof}

Now we show that the splitting reduction is a gap-preserving local reduction.
\begin{lemma}\label{lem:si-reduction}
  Let $(\caI,\epsilon,f)$ be an input of $\CSP(\caA)$ and let $(\caI',\epsilon',f') = \Call{Split}{\caI,\epsilon,f}$.
  If $(\caI',\epsilon',f')$ is testable with $q(\epsilon')$ queries, then $(\caI,\epsilon,f)$ is testable with $q(O(\epsilon))$ queries.
\end{lemma}
\begin{proof}
  We show that the splitting reduction is a linear reduction.

  Let $\caI = (V, \set{A_x}, \set{R_{xy}}, \bw)$ and $\caI' = (V',\set{A'_x}, \set{R'_{xy}}, \bw')$ be the original instance and the reduced instance, respectively.

  In the reduction, every variable $x$ of $ V$ is ultimately split into variables $x_1,\ldots,x_{k_x}$ from $V'$ and the domain $\bbA_x$ is replaced by subdirectly irreducible domains $\bbA_x^1, \ldots,  \bbA_{x}^{k_x}$ corresponding to these variables such that $\bbA_x$ is isomorphic to a subdirect product of these new domains. Since each of the domains has size bounded by $|A|$, then $k_x \le |A|$  for all $x \in V$ and so after completely splitting $\bbA_x$ into the $k_x$ factors, we have that $\bw(x) \le 2^{|A|}\bw'(x_i)$ for each $i \in [k_x]$.  We also have that $\sum_{i \in [k_x]}\bw'(x_i) = \bw(x)$ for each $x \in V$.


  We can determine the value of $f'(x_i)$, where $x_i$ is added when splitting the variable $x$,
  we only need to know the value of $f(x)$.

  If $f$ satisfies $\caI$, then $f'$ satisfies $\caI'$ by Lemma~\ref{lem:subdirectly-irreducible}.
  Suppose that $f'$ is $\epsilon/(2^{|A|})$-close to satisfying $\caI'$
   and let $g'$ be a satisfying assignment for $\caI'$ with $\dist(f',g') \leq \epsilon/(2^{|A|})$.
  Because the tuple $(g'(x_1),\ldots,g'(x_{k_x}))$ is in $A_x$, we can naturally define an assignment $g$ for $\caI$ by setting $g(x) = (g'(x_1),\ldots,g'(x_{k_x})) \in A_x$.
  Then $g$ is a satisfying assignment from Lemma~\ref{lem:subdirectly-irreducible}.
  Moreover, 
  \[
    \dist(f,g) = \sum_{x \in V: \exists i \in [k_x], g'(x_i) \neq f'(x_i)} \bw(x)
    \leq
    \sum_{x \in V}
    \sum_{i \in [k_x]: g'(x_i) \neq f'(x_i)}2^{|A|}\bw'(x_i)
    = 2^{|A|}\dist(f',g') \le  \epsilon.
  \]

  To summarize, the splitting reduction is a gap-preserving local reduction with $t(n) = |A|$, $c_1 = 1$, and $c_2 = 2^{|A|}$.
\end{proof}

\subsection{Isomorphism reduction}

By applying the factoring reduction and then the splitting reduction to an instance of $\CSP(\caA)$ we end up with an instance whose domains are either trivial or subdirectly irreducible and prime.
For such an instance, we have the following property.
\begin{lemma}\label{subprime}
  Let $\caI = (V, \set{A_x}, \set{R_{xy}}, \bw)$ be an instance of $\CSP(\caA)$ such that $|V|>1$ and such that every domain is either trivial or is subdirectly irreducible and prime.
  Then, for each variable $x \in V$,   there is at least one variable $y \ne x$ so that $\theta_{xy}=0_{A_x}$ and for such variables $y$, the relation $R_{yx}$ is the graph of a surjective homomorphism from $\bbA_y$ to $\bbA_x$.
\end{lemma}
\begin{proof}
If $|A_x| = 1$ then the result follows trivially.  Otherwise, we have that the congruence $\mu_x = \bigwedge_{y\ne x} \theta_{xy}$ of $\bbA_x$ is equal to $0_{A_x}$, since $\bbA_x$ is prime.  But, since this algebra is subdirectly irreducible, it follows that for some $y \ne x$, $\theta_{xy} = 0_{A_x}$.  Since $R_{yx}$ is a thick mapping with $\theta_{xy} = 0_{A_x}$ it follows that $R_{yx}$ is the graph of a surjective homomorphism from $\bbA_y$ to $\bbA_x$.
\end{proof}

Let $\caI = (V, \set{A_x}, \set{R_{xy}}, \bw)$ be an instance of $\CSP(\caA)$ with $|V|>1$ and with the property that every domain is either trivial or is subdirectly irreducible and prime.
Define the  relation $\sim$ on $V$ by $x \sim y$ if and only if the relation $R_{xy}$ is the graph of an isomorphism from $\bbA_x$ to $\bbA_y$.  Using the 2-consistency of $\caI$, the relation $\sim$ is naturally an equivalence relation on $V$.  The following corollary to Lemma~\ref{subprime} establishes that unless all of the domains of $\caI$ are trivial, the relation $\sim$ is non-trivial.

\begin{corollary}
  For $\caI = (V, \set{A_x}, \set{R_{xy}}, \bw)$ an instance of $\CSP(\caA)$ as in Lemma~\ref{subprime}, if $x \in V$ is such that the domain $A_x$ has maximal size and has at least two elements, then there is some $y \in V$ with $x \ne y$ and $x\sim y$.
\end{corollary}

\begin{proof}
  If $A_x$ has maximal size and has at least two elements, then let $y \in V$ be a variable such that $x \ne y$ and $R_{yx}$ the graph of a surjective homomorphism from $\bbA_y$ to $\bbA_x$.  Since $A_x$ has maximal size, it follows that $|A_y| = |A_x|$ and so $R_{yx}$ is the graph of an isomorphism from $\bbA_y$ to $\bbA_x$.
\end{proof}

For a variable $x \in V$, let $[x] := x/\sim$ denote the $\sim$-class of $V$ that $x$ belongs to.
Let $S \subseteq V$ be an arbitrary complete system of representatives of this equivalence relation and for any $\sim$-class $u$, let $s(u) \in V$ be the unique element $x \in S$ such that $x\in u$.
In particular $[s(u)] = u$ holds.

Given an assignment $f$ for $\caI$, we can test the input $(\caI, \epsilon, f)$ in two steps.
First, we test whether the values of $f$ in the $\sim$-classes of $V$ are consistent using  a consistency algorithm (Algorithm~\ref{alg:consistency}) and then we test the input obtained by contracting the $\sim$-classes using Algorithm~\ref{alg:isomorphism-reduction}.
Explanations of these two steps are contained in the next two subsections.

\subsubsection{Testing $\sim$-consistency}
We say that the input $(\caI, \epsilon, f)$ is \emph{$\sim$-consistent} if, for each $x$, $y \in V$ with $x \sim y$, $(f(x), f(y)) \in R_{xy}$.

For a $\sim$-class $u \subseteq V$ and $b \in A_{s(u)}$,
we define
\begin{align*}
\overline{\bw}(u,b) &= \sum\limits_{y \in u: f(y) = R_{s(u) y}(b) }\bw(y),\\
\overline{\bw}(u) &= \sum\limits_{b \in A_{s(u)}}\overline{\bw}(u,b),\ \  \mbox{and}\\
\overline{\bw}_{\maj}(u)&= \max_{b \in A_{s(u)}} \overline{\bw}(u,b).
\end{align*}
Note that $\overline{\bw}(u)$ is also equal to $\sum_{x \in u}\bw(x)$, the sum of the weights of the variables in $u$.
In addition, we define $\epsilon_u$ to be $(\overline{\bw}(u) - \overline{\bw}_{\maj}(u)) / \overline{\bw}(u)$ and observe that $\epsilon_u \le (|A|-1)/|A|$ since $|A_{s(x)}| \le |A|$ and so $\overline{\bw}(u)$ is the sum of at most $|A|$ terms, each of which is at most $\overline{\bw}_{\maj}(u)$.
The quantity $\epsilon_u$ represents the fraction of values, by weight, of  $f|_u$ that need to be altered in order to establish  $\sim$-consistency of the assignment over the class $u$.
Let $f_{\maj}$ be the assignment obtained from $f$ in this way.
That is, for $x \in V$, $f_{\maj}(x) = R_{s([x])x}\left(\mathop{\mathrm{argmax}}_{b \in A_{s([x])}} \overline{\bw}([x],b)\right)$. 

We need the following simple proposition to analyze our algorithm.
\begin{proposition}\label{prp:simple-prob}
  Let $X$ be a random variable taking values in $[0,1]$ such that $\E[X] \geq \epsilon$ for some $\epsilon \geq 0$.
  Then, $\Pr[X \geq \epsilon/2] \geq \epsilon/2$ holds.
\end{proposition}
\begin{proof}
  Let $p = \Pr[X \geq \epsilon/2]$.
  Then,
  \[
    \epsilon
    \leq
    \E[X]
    \leq
    1 \cdot p + \frac{\epsilon}{2} (1-p)
    \leq p + \frac{\epsilon}{2}.
  \]
  Hence, $p \geq \epsilon/2$ holds.
\end{proof}

In order to test $\sim$-consistency, we run Algorithm~\ref{alg:consistency}. 
\begin{algorithm}[t!]
  \caption{}\label{alg:consistency}
  \begin{algorithmic}[1]
  \Procedure{Consistency}{$\caI,\epsilon,f$}
    \State{Sample a set $U$ of $\Theta(1/\epsilon)$ $\sim$-classes of $\caI$. In each sampling, $u$ is chosen with probability $\overline{\bw}(u)$.}
    \For{each $u \in U$}
    \State{Sample a set $S$ of $\Theta(1/\epsilon)$ variables in $u$.
    In each sampling, a variable $x \in u$ is chosen with probability $\bw(x) / \overline{\bw}(u)$.}
    \If{there are two variables $x,y \in S$ with $f(y) \neq R_{xy}(f(x))$}
    \State Reject.
    \EndIf
    \EndFor
    \State Accept.
  \EndProcedure
  \end{algorithmic}
\end{algorithm}

\begin{lemma}\label{lem:consistency}
  Algorithm~\ref{alg:consistency} tests $\sim$-consistency with query complexity $O(1/\epsilon^2)$.
\end{lemma}
\begin{proof}
  It is clear that Algorithm~\ref{alg:consistency} accepts if $f$ is $\sim$-consistent and the query complexity is $O(1/\epsilon^2)$.
Suppose that $f$ is $\epsilon$-far from $\sim$-consistency, which means that $\dist(f,f_{\maj}) \geq \epsilon$.
  Then, we have $\E[\epsilon_u] = \sum\limits_{u:\sim\text{-class}}(\overline{\bw}(u)) \epsilon_u\geq \epsilon$, where in the calculation of the expectation, a $\sim$-class $u$ is chosen with probability $\overline{\bw}(u)$.
  Note that $\epsilon_u \in [0,1]$ for every $\sim$-class $u$  and so we can apply  Lemma~\ref{prp:simple-prob}, to conclude that we sample a $\sim$-class $u$ with $\epsilon_u \geq \epsilon/2$ with probability at least $\epsilon/2$.
  Hence, the probability that $U$ contains a $\sim$-class $u$ with $\epsilon_u \geq \epsilon/2$ is at least $1 - (1-\epsilon/2)^{\Theta(1/\epsilon)} \geq 5/6$ by choosing the hidden constant large enough.
  For a $\sim$-class $u$ with $\epsilon_u \geq \epsilon/2$,
  the probability that we find two vertices $x,y \in u$ with $f(y) \neq R_{xy}(f(x))$ in $S$ is at least
  \[
  1 - (1-\epsilon_u)^{\Theta(1/\epsilon)} - (\epsilon_u)^{\Theta(1/\epsilon)}
  \ge 1 - (1-\epsilon/2)^{\Theta(1/\epsilon)} - ((|A|-1)/|A|)^{\Theta(1/\epsilon)}
  \]
  since $\epsilon_u \ge \epsilon/2$ for this class $u$ and, as noted earlier, $\epsilon_u \le (|A|-1)/|A|$ for every class $u$.
  By choosing the hidden constant large enough we can ensure that this value is at least $5/6$.
  By combining these bounds, we obtain two vertices $x,y$ with $f(y) \neq R_{xy}(f(x))$ with probability at least $2/3$.
\end{proof}

\subsubsection{Isomorphism reduction}
Using Algorithm~\ref{alg:consistency}, we can reject an input $(\caI,\epsilon,f)$ if it is far from satisfying $\sim$-consistency.
In this subsection we will consider a reduction from $(\caI,\epsilon,f)$ to another input $(\caI',\epsilon',f')$ assuming that $(\caI,\epsilon,f)$ is close to satisfying $\sim$-consistency.

\begin{algorithm}[t!]
  \caption{}\label{alg:isomorphism-reduction}
  \begin{algorithmic}[1]
  \Procedure{Isomorphism}{$\caI,\epsilon,f$}
    \For{each $\sim$-class $u$}
      \State{Sample a variable $x \in u$ with probability $\bw(x) / \overline{\bw}(u)$, and let $x_u$ be the sampled variable.}
      \State{$V' \leftarrow V' \cup \set{u}$.}
      \State{$A'_{u} \leftarrow A_{s(u)}$.}
      \State{$\bw'(u) \leftarrow \overline{\bw}(u)$.}
      \State{$f'(u) \leftarrow R_{x_u s(u)}(f(x_u))$.}
    \EndFor
    \For{each pair $(u,u')$ of  $\sim$-classes}
      \State{$R'_{u u'} \leftarrow R_{x_u x_{u'}}$.}
    \EndFor
    \State{\Return $((V',\set{A'_x}, \set{R'_{xy}}, \bw'),\epsilon/2,f')$.}
  \EndProcedure
  \end{algorithmic}
\end{algorithm}

Our reduction, as described in Algorithm~\ref{alg:isomorphism-reduction}, contracts the variables in each $\sim$-class to a single variable from that class.
It should be clear that since the instance $\caI$ of $\CSP(\caA)$ is assumed to be 2-consistent, the reduction will produce another 2-consistent instance $\caI'$ of $\CSP(\caA)$.  As the next lemma shows, unless the domains of $\caI$ all have size one, some of the domains of $\caI'$ will no longer be prime.

\begin{lemma}\label{shrink}
  Let $(\caI, \epsilon, f)$ be an input of $\CSP(\caA)$ for which domains of $\caI$ are either trivial or prime and subdirectly irreducible and let $(\caI',\epsilon',f') = \Call{Isomorphism}{\caI,\epsilon,f}$.  If some domain of $\caI$ has more than one element, then any domain of $\caI'$  of maximal size will  not be prime, unless $\caI'$ has only one variable.
\end{lemma}
\begin{proof}
  Suppose that $\caI'$ has more than one variable.  This is equivalent to there being more than one $\sim$-class for $\caI$.
  Let $x$ be a variable of $\caI'$ with $|A_x|$ of maximal size and let $y$ be any other variable of $\caI'$. Note that according to the construction of $\caI'$ from $\caI$, both $x$ and $y$ are also variables of $\caI$ with $x \not\sim y$.  Furthermore, $|A_x|$ has maximal size amongst all of the domains of $\caI$ and so the relation $R_{yx}$ cannot be the graph of a surjective homomorphism from $\bbA_y$ to $\bbA_x$.  If it were, then it would be the graph of an isomorphism, contradicting that $x \not\sim y$.  Thus the congruence $\theta_{xy} \ne 0_{A_x}$.  Since $\bbA_x$ is subdirectly irreducible it follows that $\mu_x = \bigwedge_{y\ne x} \theta_{xy}$ is also not equal to $0_{A_x}$ and so $A_x$ is not prime in $\caI'$.
\end{proof}


\begin{lemma}\label{lem:isomorphism-reduction}
  Let $(\caI, \epsilon,f)$ be an input of $\CSP(\caA)$ and suppose that $f$ is $\epsilon/20$-close to satisfying $\caI$.
  Let $(\caI',\epsilon',f') = \Call{Isomorphism}{\caI,f}$.
  If $(\caI',\epsilon',f')$ is testable with $q(\epsilon')$ queries, then $(\caI,\epsilon,f)$ is testable with  $q(O(\epsilon))$ queries.
\end{lemma}
\begin{proof}
  We show that the reduction in Algorithm~\ref{alg:isomorphism-reduction} is a linear reduction.
  Let $\caI = (V, \set{A_x}, \set{R_{xy}}, \bw)$ and $\caI' = (V',\set{A'_x}, \set{R'_{xy}}, \bw')$ be the original instance and the reduced instance, respectively.

  Note that $|V'| \leq |V|$ and we can determine the value of $f'(u)$ by querying $f(x_u)$.
  Also, if $f$ satisfies $\caI$, then it is clear that $f'$ satisfies $\caI'$.

  We want to show that, if $f$ is far from satisfying $\caI$, then $f'$ is also far from satisfying $\caI'$ with high probability.
  To this end, we first show that the following quantity is small with high probability:
  \[
    \dist(f,f') := \sum_{u:\sim\text{-class}}\sum_{\substack{x \in u: \\ f'(u) \neq R_{xs(u)}(f(x))}}\bw(x).
  \]

  For a $\sim$-class $u$,
  we define
  \[
    \dist_u(f,f') := \sum_{\substack{x \in u: \\ f'(u) \neq R_{x s(u)}(f(x))}}\frac{\bw(x)}{\overline{\bw}(u)}.
  \]
  Note that we have $\dist(f,f') = \sum\limits_{u:\sim\text{-class}}\overline{\bw}(u)\dist_u(f,f')$.

  Then for any $\sim$-class $u$,
  \begin{align*}
    \E_{x_u}[\dist_u(f,f')]
    & =
    \sum_{b \in A_{s(u)}}\frac{\overline{\bw}(u,b)}{\overline{\bw}(u)}\left(1 - \frac{\overline{\bw}(u,b)}{\overline{\bw}(u)}\right)  \\
    & \leq
    \frac{\overline{\bw}_{\maj}(u)}{\overline{\bw}(u)}\left(1 - \frac{\overline{\bw}_{\maj}(u)}{\overline{\bw}(u)}\right) + \left(1 - \frac{\overline{\bw}_{\maj}(u)}{
  \overline{\bw}(u)}\right)\cdot 1 \\
    &\leq
    2\left(1- \frac{\overline{\bw}_{\maj}(u)}{\overline{\bw}(u)}\right) = 2\epsilon_u.
  \end{align*}
  Thus, $\E_{\set{x_u}_{u:\sim\text{-class}}}[\dist(f,f')]$ is equal to
  \begin{eqnarray*}
    \E_{\set{x_u}}\Bigl[\sum_{u:\sim\text{-class}}\overline{\bw}(u)\dist_u(f,f')\Bigr]
    \leq
    \sum_{u:\sim\text{-class}} 2\overline{\bw}(u)\epsilon_u
    <
    \frac{\epsilon}{10}.
  \end{eqnarray*}
  Also, for any $\sim$-class $u$, $\Var_{x_u}[\dist_u(f,f')]$ is at most
  \begin{eqnarray*}
    \E_{x_u}[\dist_u(f,f')^2]
    \leq
    \E[\dist_u(f,f')]
    \leq
    2 \epsilon_u.
  \end{eqnarray*}
    Here we have used the fact that $0 \leq \dist_u(f,f') \leq 1$.


    Since random variables $\{\dist_u(f,f')\}_{u:\sim\text{-class}}$ are independent, we have
  \begin{align*}
    \Var_{\set{x_u}}[\dist(f,f')] & =  \Var_{\set{x_u}}[\sum_{u:\sim\text{-class}}\overline{\bw}(u)\dist_u(f,f')]
    \leq
    \sum_{u:\sim\text{-class}} \overline{\bw}(u)^2 \cdot 2\epsilon_u
    \leq
    \sum_{u:\sim\text{-class}} 2 \overline{\bw}(u) \epsilon_u
    \leq
    \frac{\epsilon}{10}.
  \end{align*}
  From Chebyshev's inequality,
  we have $\Pr_{\set{x_u}}[\dist(f,f') \geq \epsilon/2] \leq \Pr_{\set{x_u}}[|\dist(f,f') - \epsilon/10| \geq 4 \cdot \epsilon/10 ] \leq  1/16$.

  Let $g'$ be a satisfying assignment for $\caI'$ closest to $f'$.
  We define an assignment $g$ for $\caI$ as $g(x) =  R_{s([x])x}g'([x])$.
  It is clear that $g$ is a satisfying assignment.
  Since we have $\dist(f,f') + \dist(f',g') \geq \dist(f,g) \geq \epsilon$,
  it follows that $\Pr[\dist(f',g') \geq \epsilon / 2] \geq 15/16$.

  To summarize, the isomorphism reduction is a gap-preserving local reduction with $t(n) \leq n$, $c_1 = 1$, and $c_2 = 2$.
\end{proof}

\begin{algorithm}[t!]
  \caption{}\label{alg:isomorphism-reduction-with-preprocessing}
  \begin{algorithmic}[1]
  \Procedure{$\mbox{Isomorphism}^\prime$}{$\caI,\epsilon,f$}
    \If{\Call{Consistency}{$\caI, \epsilon/20,f$} rejects}
    \State{Reject.}
    \Else
    \State{\Return \Call{Isomorphism}{$\caI,\epsilon,f$}}
    \EndIf
  \EndProcedure
  \end{algorithmic}
\end{algorithm}

Finally, we combine Algorithm~\ref{alg:consistency} and Algorithm~\ref{alg:isomorphism-reduction}.
to produce Algorithm~\ref{alg:isomorphism-reduction-with-preprocessing} and make use of it in the following.
\begin{lemma}\label{lem:isomorphism-reduction-with-preprocessing}
  Let $(\caI, \epsilon,f)$ be an input of $\CSP(\caA)$ and suppose that $\Call{Isomorphism}{\caI,f}$ returned another instance $(\caI',\epsilon',f')$.
  If $(\caI',\epsilon',f')$ is testable with $q(\epsilon')$ queries, then $(\caI,\epsilon,f)$ is testable with  $q(O(\epsilon))$ queries.
\end{lemma}
\begin{proof}
  Consider Algorithm~\ref{alg:isomorphism-reduction-with-preprocessing}.
  If $f$ satisfies $\caI$, then the $\sim$-consistency test always accepts, and hence we always accept with probability $2/3$ from Lemma~\ref{lem:isomorphism-reduction}.
  Suppose that $f$ is $\epsilon$-far from satisfying $\caI$.
  If $f$ is $\epsilon/20$-far from satisfying $\sim$-consistency, then the $\sim$-consistency test rejects with probability at least $2/3$.
  If $f$ is $\epsilon/20$-close to satisfying $\sim$-consistency, then we reject with probability at least $2/3$ by Lemma~\ref{lem:isomorphism-reduction}.
\end{proof}

\subsection{Putting things together}

Combining the reductions introduced so far we can design a shrinking reduction, which shrinks the maximum size of the domains of an instance of $\CSP(\caA)$.

\begin{algorithm}
  \caption{}\label{alg:domain-shrink}
  \begin{algorithmic}[1]
  \Procedure{Shrink}{$\caI,\epsilon,f$}
    \State{$(\caI,\epsilon,f) \leftarrow \Call{Factor}{\caI,\epsilon,f}$.}
    \State{$(\caI,\epsilon,f) \leftarrow \Call{Split}{\caI,\epsilon,f}$.}
    \If{\Call{$\mbox{Isomorphism}^\prime$}{$\caI,\epsilon,f$} rejects}
    \State{Reject.}
    \Else
    \State{$(\caI,\epsilon,f) \leftarrow$ the input returned by \Call{$\mbox{Isomorphism}^\prime$}{}.}
    \EndIf
    \State{$(\caI,\epsilon,f) \leftarrow \Call{Factor}{\caI,\epsilon,f}$.}
    \State{\Return $(\caI,\epsilon,f)$.}
  \EndProcedure
  \end{algorithmic}
\end{algorithm}

\begin{lemma}\label{lem:reduction}
  Let $(\caI,\epsilon,f)$ be an input of $\CSP(\caA)$, and suppose that \Call{Shrink}{$\caI,\epsilon,f$} returned another instance $(\caI',\epsilon',f')$.
  If we can test $(\caI',\epsilon',f')$ with $q(\epsilon')$ queries, then we can test $(\caI,\epsilon,f)$ with $q(O(\epsilon))$ queries.
  Moreover, the reduction reduces the maximum size of a domain of the given input, if this maximum is greater than one and the reduced instance has more than one variable.
\end{lemma}
\begin{proof}
  We note that at each step of the algorithm, the domains of the instances that are produced are no larger than the domains of the original instance.  Furthermore, if any of the domains of the  original instance has size greater than one, then it follows from Lemma~\ref{shrink} that the maximal size of the domains of the output instance will be smaller than that of the original instance, as long as the output instance has more than one variable.
\end{proof}

\begin{theorem}
  Let $\bfA$ be a structure that has majority and Maltsev polymorphisms.
  Then, $\CSP(\bfA)$ and $\eCSP(\bfA)$ are constant-query testable with one-sided error.
\end{theorem}
\begin{proof}
  By applying the shrinking reduction at most $|A|$ times, we get an instance for which every variable has a domain of size one or which has only one variable.  In either case, the testing becomes trivial.
\end{proof}


\section{Non Constant-Query Testability}\label{sec:non-cp}

In this section we consider structures $\bfA$ that do not have a majority polymorphism or do not have a Maltsev polymorphism.  As noted in the previous section, this is the same as the variety $\caV(\Alg(\bfA))$ failing to be arithmetic.  For such structures we will show that $\CSP(\bfA)$ is not constant-query testable.

From~\cite{Ho-McK} we know that for a structure $\bfA$, having both  majority and Maltsev polymorphisms is equivalent to $\caV(\Alg(\bfA))$ being congruence meet semidistributive and congruence permutable.

 First suppose  that $\caV(\Alg(\bfA))$ is not congruence meet semidistributive.
We  observe that $\CSP(\bfA')$ will be sublinear-query testable if $\CSP(\bfA)$ is, where $\bfA'$ is obtained from $\bfA$ by adding all the unary constant relations (see Lemma~5 of~\cite{Bhattacharyya:2013fa}).
Although the original proof of the Lemma only considers the unweighted Boolean case, it is straightforward to generalize it to the weighted finite domain case, and we do not repeat it here.
By adding all of the unary constant relations to $\bfA$ to produce $\bfA'$ it follows that the variety $\caV(\Alg(\bfA'))$ is idempotent and  will also not be congruence meet semidistributive.
For such a structure, it is known that testing $\CSP(\bfA')$ requires a linear number of queries~\cite{Yoshida:2014zn}, and hence testing $\CSP(\bfA)$ will also require a linear number of queries.

From the argument above, in order to complete the proof of Theorem~\ref{the:non-arith} it suffices to show that $\CSP(\bfA)$ is not constant-query testable when $\bfA$ does not have a Maltsev polymorphism (or equivalently when $\caV(\Alg(\bfA))$ is not congruence permutable).
We use the following fact. 
\begin{lemma}\label{lem:non-cp->gamma}
  Let $\bfA$ be a relational structure that does not have a Maltsev polymorphism.
  Then, there is some finite algebra $\bbB$ in $\caV(\Alg(\bfA))$ and some subuniverse $\gamma$ of $\bbB^2$ such that there are elements  $0$ and $1\in B$ with $(0,0),(0,1),(1,1) \in \gamma$ and $(1,0) \not \in \gamma$.
\end{lemma}
\begin{proof}
Since $\bfA$ does not have a Maltsev polymorphism, then $\caV(\Alg(\bfA))$ is not congruence permutable and so there is some finite algebra $\bbB \in \caV(\Alg(\bfA))$  having congruences  $\alpha$ and $\beta$ such that $\alpha \circ \beta \neq \beta \circ \alpha$.  We may assume that $\alpha \circ \beta \not\subseteq\beta \circ \alpha$ and so there will be elements $0$, $1 \in B$ with $(0,1) \in \alpha \circ \beta$ but $(1,0) \notin \alpha \circ \beta$.  Since $\alpha \circ \beta$ is a reflexive relation, then setting $\gamma = \alpha \circ \beta$ works.
\end{proof}

We now establish a super-constant lower bound for $\CSP((B;\gamma))$ for $B$ and $\gamma$ as in  Lemma~\ref{lem:non-cp->gamma}.
Although the argument is similar to a super-constant lower bound for monotonicity testing given in~\cite{Fischer:2002ev}, we present it here  for completeness.

Let $G=(V; E)$ be an undirected graph and let $M \subseteq E$ be a matching in $G$, i.e., no two edges in $M$ have a vertex in common.
Let $V(M)$ be the set of the endpoints of edges in $M$.
A matching $M$ is called \emph{induced} if the subgraph induced by $M$ contains only the edges of $M$.
A bipartite graph $G=(X,Y;E)$ is called \emph{$(s,t)$-Ruzs{\'a}-Szemer{\'e}di} if its edge set can be partitioned into at least $s$ induced matchings $M_1,\ldots,M_s$, each of size at least $t$.

\begin{lemma}[Theorem 16 of~\cite{Fischer:2002ev}]\label{lem:RS-graph}
  There exist an $(n^{\Omega(1/\log \log n)}, n/3-o(n))$-Ruzs{\'a}-Szemer{\'e}di graphs $G = (X, Y; E)$ with $|X|=|Y|=n$.
\end{lemma}

\begin{theorem}\label{the:non-cp}
Let $\bfB = (B;\gamma)$ where $\gamma$ is a binary relation such that for some $0$, $1 \in B$, $(0,1)$, $(0,0$), and $(1,1) \in \gamma$ but $(1,0) \notin\gamma$.
  Then, $\CSP(\bfB)$ is not constant-query testable.
\end{theorem}
\begin{proof}

  If $\CSP(\bfB)$ is testable with $q$ queries, then $\CSP(\bfB)$ is non-adaptively testable with $|A|^q$ queries.
  Hence, in order to show that $\CSP(\bfB)$ is not constant-query testable, it suffices to show that $\CSP(\bfB)$ is not constant-query testable non-adaptively.

  Let $G = (X,Y;E)$ be an $(s,n/3-o(n))$-Ruzs{\'a}-Szemer{\'e}di graph provided as in Lemma~\ref{lem:RS-graph}, where $s = n^{\Omega(1/\log\log n)}$.
  Then, we construct an instance $\caI = (V, \caC, \bw)$ of $\CSP(\bfB)$, where $V = X \cup Y$, $\caC = \set{\langle (x,y), \gamma  \rangle \mid (x,y) \in E}$, and $\bw(x) = 1/|V|$ for all $x \in V$.

  We use Yao's principle, which states that to establish a lower bound on the complexity of a randomized test, it is enough to present an input distribution on which any deterministic test with that complexity is likely to fail.
  Namely, we define distributions $D_P$, $D_N$ on positive (satisfying) and negative (far from satisfying) assignments, respectively.
  Our assignment distribution first chooses $D_P$ or $D_N$ with equal probability and then draws an assignment according to the chosen distribution.
  We show that every deterministic non-adaptive test with $q = o(\sqrt{s})$ queries has error probability larger than $1/3$ (with respect to the induced probability on assignments).

  We now define the distributions $D_P$ and $D_N$, as well as the auxiliary distribution $\widetilde{D}_N$.
  For $D_P$ and $D_N$, choose a random $i \in \set{1,\ldots,s}$ uniformly.
  For all variables $x \in X$ and $y \in Y$ outside of matching $M_i$, set $f(x) = 0$ and $f(y)=1$.
  For $D_P$, uniformly choose $f(x) = f(y) = 0$ or $f(x)=f(y)=1$ independently for all edges $(x,y) \in M_i$.
  For $\widetilde{D}_N$, uniformly choose $f(x) = 1-f(y) = 0$ or $f(x) = 1- f(y) = 1$ independently for all $(x,y) \in M_i$.

  Note that $D_P$ is supported only on positive assignments, but $\widetilde{D}_N$ is not supported only on negative assignments.
  However, for $n$ large enough, with probability more than $8/9$ at least $1/3$ of the constraints on the edges of $M_i$ are violated when the assignment is chosen according to  $\widetilde{D}_N$, making the assignment $\Omega(1)$-far from satisfying $\caI$.
  Denote the latter event by $A$ and define $\widetilde{D}_N|_A$, namely, $D_N$ is $\widetilde{D}_N$ conditioned on the event $A$.
  Note that for $\widetilde{D}_N$, a constraint is violated only if it belongs to $M_i$, since the matchings are induced.

  Given a deterministic non-adaptive test that makes a set $V'$ of $q$ queries, the probability that one or more of the edges of $M_i$ have both endpoints in $V'$ is at most $q^2/(4s)$ for both $D_P$ and $D_N$.
  This is because the matchings are disjoint, and the vertex set $V'$ induces at most $q^2/4$ edges of $G$.
  For $q = o(\sqrt{s})$, with probability more than $1-o(1)$, no edge of $M_i$ has both endpoints in $V'$.
  Conditioned on any choice of $i$ for which $M_i$ has no such edge, the distribution of $f|_{V'}$ is identical for both $D_N$ and $D_P$:
  every vertex outside of $M_i$ is fixed to $0$ if it is in $X$ and to $1$ if it is in $Y$, and the value of every other vertex is uniform and independent over $\set{0,1}$.
  Let $C(\phi)$ denote the set of assignments consistent with query answers $\phi:V' \to \set{0,1}$.
  Then, we have $\Pr_{D_P}[C(\phi) \mid \text{no edge in }M_i] = \Pr_{\widetilde{D}_P}[C(\phi) \mid \text{no edge in }M_i]$.
  For every tuple of answers $\phi$, the error probability under the above conditioning (with negative assignments chosen under $\widetilde{D}_N$ rather than $D_N$) is $1/2$.
  As the probability of the condition is at least $1-o(1)$, the overall error probability without the conditioning is at least $1/2-o(1)$.
  Since negative assignments are chosen under $D_N$, not $\widetilde{D}_N$, the success probability is $(1/2+o(1)) \cdot (\Pr[A]^{-1}) \leq (1/2+o(1)) \cdot 9/8 \leq 9/16+o(1)$.
  Thus, the error probability is $\geq 7/16-o(1)$.
\end{proof}







We can now prove the following theorem.
\begin{theorem}\label{the:non-arith}
   If the relational structure $\bfA$ does not have a majority polymorphism or does not have a Maltsev polymorphism, then $\CSP(\bfA)$ is not constant-query testable.
\end{theorem}

\begin{proof}
  As noted earlier, it suffices to establish  hardness under the assumption that $\caV(\Alg(\bfA))$ is not congruence meet semidistributive or that $\bfA$ does not have a Maltsev polymorphism.  The discussion prior to Lemma~\ref{lem:non-cp->gamma} handles the former case, while the combination of that Lemma with Theorem~\ref{the:non-cp} and Lemma~\ref{lem:reduction-for-csps} handles the latter case.
\end{proof}


\section{Sublinear-Query Testability}\label{sec:cd}

Let $\bfA = (A; \Gamma)$ be a relational structure having, for some $k \ge 2$ a $(k+1)$-ary near unanimity  polymorphism.
In this section, we will show that $\CSP(\bfA)$ and $\eCSP(\bfA)$ are sublinear-query testable.
From Lemma~\ref{lem:reduction-from-ecsp-to-csp}, it suffices only to consider $\CSP(\bfA)$.
We note that a sublinear-query tester for $\CSP(\bfA)$ is already known for the unweighted case~\cite{Bhattacharyya:2013fa}, and we will slightly modify their argument to handle weights.


Let $\caI = (V,\caC,\bw)$ be an instance of $\CSP(\bfA)$.
Since $\Gamma$ has a $(k+1)$-ary near-unanimity  polymorphism, we can assume that each constraint in $\caC$ has arity exactly $k$~\cite{Feder:1998iu}.
Hence, we can write $\caI $ as $(V,\set{R_S}_{S \subseteq V^k}, \bw)$ or simply $\caI = (V,\set{R_S},\bw)$.
Moreover, we can assume that $\caI$ is $k$-consistent.
Recall that a subset of variables $S \subseteq V$ of size $k$ is said to be \emph{violated} with respect to an assignment $f:V \to A$ if $f|_S \not \in R_S$.
We have the following fact.

\begin{lemma}\label{lem:many-violated-sets}
  If an assignment $f:V \to A$ is $\epsilon$-far from satisfying $\caI$,
  then there is a family $\caS$ of disjoint violated sets (that are contained in $V$) such that $\sum_{S \in \caS}\prod_{x \in S}\bw(x) \geq \frac{\epsilon^k}{k2^k  n^{k-1}}$.
\end{lemma}
\begin{proof}
  Let $V'$ be the set of variables $x$ such that $\bw(x) \geq \epsilon/(2n)$ and let $\caU$ be the family of violated sets $S \subseteq V'$ of size $k$.
  We say that a subset $H \subseteq V'$ is a \emph{hitting set} of $\caU$ if, for any subset $S \in \caU$, $H$ and $S$ intersect.
  We first observe that, for any hitting set $H \subseteq V'$ of $\caU$, the partial assignment $f|_{V' \setminus H}$ is extendable to a satisfying assignment.
  Indeed, if $f|_{V' \setminus H}$ is not extendable,
  then there must be a variable set $S \subseteq V' \setminus H$ with $|S| = k$ and $f|_S \not \in R_S$ from the Helly property.
  However, such a set $S$ must be contained in $\caU$,  a contradiction.

  Since $f$ is $\epsilon$-far and every variable in $V \setminus V'$ has a weight at most $\epsilon/(2n)$,
  we have $\sum_{x \in H}\bw(x) \geq \epsilon - \epsilon / (2n) \cdot n = \epsilon/2$.
  Then, we can take a family $\caS$ of disjoint violated sets such that $\sum_{S \in \caS}\sum_{x \in S}\bw(x) \geq \epsilon/2$.
  In particular, this means that
  \[
    \sum_{S \in \caS}\prod_{x \in S}\bw(x) \geq
    \Bigl(\frac{\epsilon}{2n}\Bigr)^{k-1}\sum_{S \in \caS} \max_{x \in S}\bw(x)
    \geq
    \frac{1}{k}\Bigl(\frac{\epsilon}{2n}\Bigr)^{k-1}\sum_{S \in \caS} \sum_{x \in S}\bw(x)
    \geq
    \frac{\epsilon}{2k} \Bigl(\frac{\epsilon}{2n}\Bigr)^{k-1}
    =
    \frac{\epsilon^k}{k2^k  n^{k-1}}.
    \qedhere
  \]
\end{proof}

Now we establish the main theorem of this section.
\begin{theorem}
  Let $\bfA$ be a relational structure that has a $(k+1)$-ary near unanimity polymorphism for some $k \geq 2$.
  Then, $\CSP(\bfA)$ (and hence $\eCSP(\bfA)$) is sublinear-query testable with one-sided error.
\end{theorem}

\begin{proof}
  First, we describe our algorithm.
  Let $(\caI = (V,\caC,\bw),\epsilon,f)$ be an input of $\CSP(\bfA)$.
  We query each variable $x \in V$ with probability $q \cdot \bw(x)$,
  where $q = \Theta((\frac{k2^k  n^{k-1}}{\epsilon^k})^{1/k}) = \Theta(\frac{2}{\epsilon}k^{1/k}n^{(k-1)/k})$.
  If we query more than $100q$ times along the way, we immediately stop and accept.
  Suppose that the number of queries is at most $100q$.
  Then, we reject if there is some subset $S \subseteq V$ of size $k$ such that $f|_S \not \in R_S$,
  and we accept otherwise.
  The query complexity is $100q$, which is sublinear in $n$.

  It is easy to see that the algorithm always accepts if $f$ is a satisfying assignment (no matter whether we stopped as we have queried more than $100q$ times).

  Now, we see that the algorithm rejects with high probability when the input is $\epsilon$-far.
  From Markov's inequality, the query complexity is at most $100q$ with probability at least $\frac{99}{100}$.

  From Lemma~\ref{lem:many-violated-sets},
  there is a family $\caS$ of disjoint violated sets such that $\sum_{S \in \caS}\prod_{x \in S}\bw(x) \geq  \frac{\epsilon^k}{k2^k  n^{k-1}}$.
  Note that for each violated variable set $S \in \caS$,
  the probability that we do not find $S$ is $1 - \prod_{x \in S}(q \cdot \bw(x) ) = 1 - q^k \prod_{x \in S}\bw(x)$.
  Thus,
  because violated sets in $\caS$ are disjoint,
  the probability that we do not find any violated variable set  is at most
  \begin{align*}
    \prod_{S \in \caS}\left(1 - q^k\prod_{x \in S}\bw(x)\right)
    & \leq
    \prod_{S \in \caS}\exp\left(-q^k\prod_{x \in S}\bw(x)\right)
    =
    \exp\left(-q^k\sum_{S \in \caS}\prod_{x \in S}\bw(x)\right)
    \leq
    \exp\left(- q^k \frac{\epsilon^k}{k2^k  n^{k-1}} \right)
  \end{align*}
  If we choose the constant hidden in $q$ large enough,
  the probability above is bounded by $\frac{1}{100}$.
  Thus, with probability at least $\frac{98}{100}$,
  we reject the instance.
\end{proof}



\newcommand{\leval}{\mathsf{Lattice\mbox{-}Eval}}
\newcommand{\peval}{\mathsf{Pent\mbox{-}Eval}}

\newcommand{\eq}{\mathrm{Eq}}

\section{Non Sublinear-Query Testability}\label{sec:non-cd}

Let $\bfA$ be a relational structure that does not have, for any $k \ge 2$, a $(k+1)$-ary near unanimity polymorphism.
In this section, we show that $\eCSP(\bfA)$ is not sublinear-query testable with one-sided error.

Using Barto's proof of Z\'adori's Conjecture~\cite{Barto:2013uo}, we know that a finite relational structure $\bfA$ has a $(k+1)$-ary near unanimity polymorphism for some $k \ge 2$ if and only if the variety $\caV(\Alg(\bfA))$ is congruence distributive.  Furthermore, this condition is equivalent to $\caV(\Alg(\bfA))$ being congruence meet semidistributive and congruence modular~\cite{Ho-McK}.


From the argument in Section~\ref{sec:non-cp} we know that if $\caV(\Alg(\bfA))$ is not congruence meet semidistributive, then testing $\CSP(\bfA)$ and hence $\eCSP(\bfA)$ requires a linear number of queries.
Hence, to establish the main result of this section, Theorem~\ref{the:non-cm-ecsp}, it suffices to show that $\eCSP(\bfA)$ is not sublinear-query testable with one-sided error when $\caV(\Alg(\bfA))$ is not congruence modular.
The results in this section make use of ideas developed in~\cite{BovaChenValeriote13-generic} and~\cite{Chen:2015wh}.

A \emph{lattice} is an algebra $(L; \wedge, \vee)$
where $L$ is a domain, and each of the operations $\wedge, \vee$ is idempotent, commutative, and associative; and,
the absorption law $a \wedge (a \vee b) = a \vee (a \wedge b) = a$ holds.
A lattice naturally induces a partial order $\leq$
defined by $a \leq b$ if and only if $a \wedge b = a$.
A lattice is \emph{distributive} if
it satisfies the identity
$x \wedge (y \vee z) = (x \wedge y) \vee (x \wedge z)$.
A lattice is \emph{finite} if its domain is finite,
and is \emph{non-trivial} if its domain has size
strictly larger than $1$.
It is known that a finite lattice has a
\emph{bottom element} $\bot$
and a
\emph{top element} $\top$
such that for each element $a$, it holds that
$\bot \leq a \leq \top$.


Let $\bbL = (L; \meet, \join)$ be a finite lattice and let $D > 1$ be a constant.
We define $\leval(\bbL, D)$ to be the assignment problem
where an instance consists of the following:
\begin{itemize}
\itemsep=0pt
\item A circuit, on a variable set $V$,
over basis $\{ \wedge, \vee \}$
of depth less than $D + D \log_2 |V|$.
Here, $\wedge$ and $\vee$ are always assumed to have fan-in $2$.

\item An element $\ell \in L$.

\item A weight function $\bw: V \to [0,1]$.

\end{itemize}
The assignments associated to an instance are
the mappings from $V$ to $L$; such a mapping
$f: V \to L$ is considered to be \emph{satisfying}
if $C(f) \geq \ell$.
Here and in general,
when $C$ is a circuit on variable set $V$
and $f$ is an assignment defined on $V$, we use $C(f)$
to denote the result of evaluating $C$ under $f$.

\begin{lemma}\label{lem:non-cm-reduction-1}
There exists a constant $D > 1$ such that
testing $\leval((\set{ 0, 1 }, \set{\wedge, \vee}), D)$
(where $\ell = 1$) with one-sided error requires a linear number of queries.
\end{lemma}

\begin{proof}
We can view the stated problem as that of testing
assignments to logarithmic depth monotone circuits.
It is known~\cite{Valiant:1984eu} that
there are such circuits for the \emph{majority} function.
Hence, it suffices to argue that testing the (uniformly weighted) majority function with one-sided error requires a linear number of queries.

To see this, let us think about the behavior of a one-sided error tester $T$, given the all-zero assignment $f:V \to \set{0,1}$.
Notice that $f$ is $1/2$-far from satisfying the majority function.
Hence, when $\epsilon < 1/2$, the tester $T$ must reject $f$ with probability at least $2/3$.

Suppose that $T$ has queried variables in $S \subseteq V$ with $|S| < n/2$.
Then, the assignment $f'$ with $f'(x) = 0$ for every $x \in S$ and $f'(x) = 1$ for every $x \in V \setminus S$ is consistent with what $T$ has seen.
However, $f'$ satisfies the majority function, and hence $T$ cannot reject $f$.
This means that the query complexity of $T$ must be $\Omega(n)$.
\end{proof}

\begin{lemma}\label{lem:non-cm-reduction-2}
  Let $\bbL$ be a non-trivial finite lattice.
  For each constant $D > 1$, there exists a constant $D' > 1$ such that there is a linear reduction from the problem $\leval((\{ 0, 1 \}; \wedge, \vee), D)$ (where $\ell = 1$) to the problem $\leval(\bbL, D')$.
\end{lemma}

\begin{proof}
We first consider the case where $L$ is a distributive lattice.
It is well-known and straightforward to verify that each finite distributive
lattice embeds into a finite power of the two-element lattice.
We thus view $L$ as a sublattice of a finite power of the two-element lattice,
and in particular assume that the domain of $L$
is a subset of $\{ 0, 1 \}^k$.
We may further assume that the bottom element of $L$ is $(0, \ldots, 0)$.

Let $\caI = (C, \ell = 1, \bw)$ be an instance of
the problem $\leval((\{ 0, 1 \}; \wedge, \vee), D)$.
Fix $t \in \{ 0, 1 \}^k$ to be an element of $L$ such that
$t \neq (0, \ldots, 0)$.
The instance created is
$\caI' = (C, t, \bw)$.
An assignment $f: V \to \{0,1\}$ is mapped to
the assignment $f': V \to \{ 0, 1 \}^k$
defined by $f'(x) = (0, \ldots, 0)$ if $f(x) = 0$,
and $f'(x) = t$ if $f(x) = 1$.
It is straightforward to verify that $C(f')$ is equal to $0$ or $t$
depending on whether or not $C(f)$ is equal to $0$ or $1$, respectively.
Hence, if $f$ is a satisfying assignment of $\caI$,
then $f'$ is a satisfying assignment of $\caI'$.

We claim that if $g'$ is a satisfying assignment of $\caI'$ such that
$\dist_{\caI'}(f',g') < \epsilon$,
then there exists a satisfying assignment $g$ of $\caI$
such that
$\dist_{\caI}(f,g) < \epsilon$.
This implies that the constant $c_1$ in the definition of reduction
can be taken as $c_1 = 1$.
In particular, define $g: V \to \{ 0, 1 \}$
so that $g(x) = 0$ if $g'(x) = (0, \ldots, 0)$,
and $g(x) = 1$ otherwise.
The inequality
$\dist_{\caI}(f,g) < \epsilon$
holds as a consequence of the fact that
(for each $x \in V$) $f(x) \neq g(x)$ implies $f'(x) \neq g'(x)$;
this fact can be verified by a case analysis of the possible values
$(0, 1)$, $(1,0)$ for $(f(x),g(x))$.
Now fix $r$ to be an index such that the $r$th entry $t_r$
of the tuple $t$ is equal to $1$.  (We extend this subscript notation
to assignments mapping to $L$ in the natural fashion.)
Since $g'$ is a satisfying assignment, $C(g') \geq t$,
implying that $C(g'_r) = C(g')_r \geq t_r = 1$; since
$g \geq g'_r$ (by definition of $g$), $C(g) = 1$
and
the assignment $g$ is satisfying, with respect to $\caI$.

We now consider the case where $L$ is a non-distributive lattice.
Define $s(x,y,z)$ to be $(x \wedge y) \vee (x \wedge z)$
and define $s'(x,y,z)$ to be $x \wedge (y \vee z)$.
Under any assignment to the variables $\{ x, y, z \}$, it holds that
$s \leq s'$.
Fix values $a, b, c \in L$ such that $s(a,b,c) \neq s'(a,b,c)$;
such values exist
by the assumption that $L$ is non-distributive.
For any elements $d, d' \in L$, define
$[d,d']$ as $\{ e \mid d \leq e \leq d' \}$.
It is straightforward to verify that
$[s(a,b,c), s'(a,b,c)] \neq L$.
Define $\bbL^-$ to be the sublattice of $\bbL$
with domain $[s(a,b,c), s'(a,b,c)]$.

By induction, it suffices to show that,
for each $D > 1$, there exists $D' > 1$ such that
there is a reduction from $\leval(\bbL^-,D)$ to $\leval(\bbL, D')$.
Let $\caI^- = (C^-(x_1, \ldots, x_n), \ell, \bw^-)$ be an instance of
$\leval(\bbL^-,D)$ where we use $V = \{ x_1, \ldots, x_n \}$
to denote the variable set of $C^-$,
and let $f^-: V \to L^-$ be an assignment of $\caI^-$.
The instance $\caI$ produced by the reduction is
$(C, \ell, \bw)$, where
$C(z_1, z_2, z_3, x_1, \ldots, x_n)$
is the circuit on variable set $V \cup \{ z_1, z_2, z_3 \}$
defined as $C^-(x^*_1, \ldots, x^*_n)$;
each $x^*_i$ is equal to the circuit
$(x_i \vee s(z_1, z_2, z_3)) \wedge s'(z_1, z_2, z_3)$.
Observe that the depth of the created circuit $C$
is equal to that of $C^-$, plus a constant.
The weight function $\bw$ is defined as
$\bw(z_1) = \bw(z_2) = \bw(z_3) = 1/4$
and $\bw(x) = \bw^-(x) / 4$ for each $x \in V$.
The assignment $f: V \cup \{ z_1, z_2, z_3 \} \to L$
produced by the reduction is the extension of $f^-$
that maps $(z_1, z_2, z_3)$ to $(a,b,c)$.

Clearly, if $f^-$ is a satisfying assignment of $\caI^-$,
then $f$ is a satisfying assignment of $\caI$,
since $C^-(f^-) = C(f)$.
We claim that the constant $c_1$ in the definition of reduction
can be taken as $1/4$.
Suppose that $g$ is a satisfying assignment of $\caI$
such that $\dist_{\caI}(f,g) < \epsilon / 4$.
Then $g$ must be equal to $f$ on $\{ z_1, z_2, z_3 \}$, since
each of those variables has weight $1/4$.
Define $g^-(x) = (g(x) \vee s(a,b,c)) \wedge s'(a,b,c)$,
for each $x \in V$; by definition of $C$, it holds that
$g^-$ is a satisfying assignment of $\caI$.
Observe that $f(x) = g(x)$ implies that
$f^-(x) = g^-(x)$,
since $f^-(x) = (f(x) \vee s(a,b,c)) \wedge s'(a,b,c)$.
We conclude that $\dist_{\caI}(f^-,g^-) < \epsilon$, establishing the claim.
\end{proof}

Let $A$ be a set.
Recall that for binary relations $\theta$ and $\theta'$  on $A$,
we use $\theta \circ \theta'$ to denote their relational product.
We use $\eq(A)$ to denote the lattice of equivalence relations on $A$,
and we use $0_A = \{ (a, a) \mid a \in A \}$
and $1_A = A^2$ to denote the bottom and top elements of
$\eq(A)$, respectively.
We define a \emph{pentagon} to be a finite relational structure $\bfP$
over the signature (or relational structure language) $\{ \alpha, \beta, \gamma \}$
containing three binary relation symbols such that
$\alpha^{\bfP}$, $\beta^{\bfP}$, and $\gamma^{\bfP}$
are equivalence relations on $P$, and the following conditions hold
in $\eq(P)$:
$\alpha^{\bfP} \leq \beta^{\bfP}$,
$\beta^{\bfP} \wedge \gamma^{\bfP} = 0_P$,
$\beta^{\bfP} \circ \gamma^{\bfP} = 1_P$,
and
$\alpha^{\bfP} \vee \gamma^{\bfP} = 1_P$.
The domain $P$ of a pentagon $\bfP$
can be naturally decomposed as a direct product $P = B \times C$
in such a way that $\beta^{\bfP}$ and $\gamma^{\bfP}$
are the kernels of the projections of $P$ onto $B$ and $C$, respectively.
Then, via the equivalence relation $\alpha^{\bfP}$,
each element $b \in B$ induces an equivalence relation
$\alpha^{\bfP}_b =
\{ (c, c') \in C \times C \mid ((b,c), (b,c')) \in \alpha^{\bfP} \}$
on $C$.
For each pentagon $\bfP$, we define
$\bbL(\bfP)$ to be the lattice which is
the sublattice of $\eq(C)$ generated by
the equivalence relations $\alpha^{\bfP}_b$ (over $b \in B$).

To each pentagon $\bfP$,
we associate a $2$-sorted relational structure,
denoted by $\bfP_2$,
which has $B_\bfP$ and $C_\bfP$ as first and second domain,
respectively; here, $B_\bfP$ and $C_\bfP$ denote the sets
in the decomposition of the domain $P$ as described above.
The structure $\bfP_2$ is defined on signature $\{ R \}$
and has $R^{\bfP_2} =
\{ (b, c, c') \in B_\bfP \times C_\bfP \times C_\bfP \mid (c, c') \in \alpha^{\bfP}_b \}$.
The definition of $\bfP_2$ comes from~\cite{BovaChenValeriote13-generic}.
In forming conjunctive queries over this signature $\{ R \}$
each variable has a sort (first or second) associated with each variable;
an atom $R(x, y, y')$ may be formed if $x$ is of the first sort
and $y$ and $y'$ are of the second sort.


When $\bfP$ is a pentagon,
we define $\peval(\bfP)$
to be the assignment problem where an instance consists of
the following:
\begin{itemize}

\item A pp-formula $\phi(X, Y)$
on the signature $\{ R \}$
of $\bfP_2$, where the variables in the sets $X$ and $Y$
are of the first and second sort, respectively.
(We assume $X \cap Y = \emptyset$.)

\item A weight function $\bw: X \cup Y \to [0,1]$.

\end{itemize}
The assignments associated to an instance
are the mappings $h = h_1 \cup h_2$
where $h_1$ is a mapping from $X$ to $B_{\bfP}$
and $h_2$ is a mapping from $Y$ to $C_{\bfP}$.
Such a mapping is satisfying if it causes $\phi$
to evaluate to true over $\bfP_2$.

\begin{lemma}\label{lem:non-cm-reduction-3}
Let $\bfP$ be a pentagon.
For each $D \geq 1$, there exists a linear reduction
from the problem $\leval(\bbL(\bfP),D)$ to $\peval(\bfP)$.
\end{lemma}

\begin{proof}
We first observe that there exists a constant $E$
such that $\leval(\bbL(\bfP),D)$ linearly reduces to
the special case of $\leval(\bbL(\bfP),D+E)$
where the assignment must map to
the set of generators $G = \{ \alpha_b^{\bfP} \mid b \in B_{\bfP} \}$
of $L(\bfP)$.
Let $s(x_1, \ldots, x_q)$ be a fixed lattice term
that maps $G^q$ surjectively onto $L(\bfP)$.
The reduction, on $(T(x_1, \ldots, x_m), \ell, \bw), f$,
produces
$$(T(s(x_1^1, \ldots, x_1^q), \ldots, s(x_m^1, \ldots, x_m^q)), \ell, \bw'), f'.$$
Here, the assignment $f'$ is defined so that,
for each $i \in [m]$, it holds that
$s(f'(x_i^1), \ldots, f'(x_i^q)) = f(x_i)$;
the function $\bw'$ is defined by $\bw(x_i^j) = w(x_i)/q$
for all $i \in [m]$, $j \in [q]$.
It is straightforward to verify that the constant $c_1$
in the definition of reduction can be taken as $1/q$;
the key point is that, in order to modify a value
$s(x_i^1, \ldots, x_i^q)$ of the new instance,
which corresponds to the input $x_i$ to $T$,
it is necessary to change at least one of the values
$x_i^1, \ldots, x_i^q$, whose weight is $1/q$ times the
weight $w(x_i)$.

It thus suffices to give a linear reduction from
this special case of $\leval(L(\bfP), D+E)$
to $\peval(\bfP)$, which is what we now do.
Let $(T(x_1, \ldots, x_m), \ell, \bw), f$ be an input
to the first problem.
We make use of a construction in the literature
(introduced in~\cite{BovaChenValeriote13-generic}
and also employed in~\cite{Chen:2015wh})
which allows us to create, from the circuit $T(x_1, \ldots, x_m)$,
a pp-formula $\phi_T(x_1, \ldots, x_m, y_1, y_2)$ (over $\bfP_2$)
such that, for all $b_1, \ldots, b_m \in B_{\bfP}$
and all $c_1, c_2 \in C_{\bfP}$, we have that
$\phi_T(b_1, \ldots, b_m, c_1, c_2)$ holds on $\bfP_2$
if and only if $(c_1, c_2)$ is in the equivalence relation
given by $T(\alpha^{\bfP}_{b_1}, \ldots, \alpha^{\bfP}_{b_m})$,
where here it is understood that $T$ is evaluated
in the lattice $\bbL(\bfP)$.
For the sake of completeness, we briefly specify the version
of the construction used here.
The construction is defined inductively.
When $T = x_i$, we have $\phi_T(x_1, \ldots, x_m, y_1, y_2) = R(x_i, y_1, y_2)$.
When $T = T_1 \wedge T_2$, we have
$\phi_{T_1}(x_1, \ldots, x_m, y_1, y_2) \wedge
 \phi_{T_2}(x_1, \ldots, x_m, y_1, y_2)$.
When $T = T_1 \vee T_2$,
set $u = |C_{\bfP}|$.
Let $z_{0,2}$ and $z_{i,j}$, where $i = 1, \ldots, u$ and
$j = 1, 2$, be variables of the second sort,
and identify $y = z_{0,2}$ and $y' = z_{u,2}$.
Then $\phi_T$ is the formula
$\exists z_{1,1} \exists z_{1,2} \ldots \exists z_{u-1,1} \exists z_{u-1,2} \exists z_{u,1} \bigwedge_{i=1}^u (\phi_{T_1}(x_1, \ldots, x_m, z_{i-1,2},z_{i,1}) \wedge \phi_{T_2}(x_1, \ldots, x_m, z_{i,1}, z_{i,2}))$,
where here all of the variables $z_{\cdot,\cdot}$ are existentially quantified, other than $y$ and $y'$.

The reduction produces the input
$(\bigwedge_{i \in [u^2]} \phi_T(x_1, \ldots, x_m, y^1_i, y^2_i), \bw'), f'$, where
$u = |C_{\bfP}|$, and
$f'$ and $\bw'$ are described as follows.
Define $f'(x_i)$ so that $f(x_i) = \alpha^{\bfP}_{f'(x_i)}$,
and define $f'$ on the $(y^1_i, y^2_i)$ so that
$\ell = \{ (f'(y^1_i), f'(y^2_i)) \mid i \in [u^2] \}$.
Define $\bw'(y^1_i) = \bw'(y^2_i) = 1 / (2u^2 + 1)$,
and define $\bw'(x_i) = \bw(x_i) / (2u^2 + 1)$.
The reduction works with $c_1 = 1/ (2u^2 + 1)$, for
if $f'$ is within distance $c_1 \epsilon$ of a satisfying
assignment, then the satisfying assignment must be equal
on the variables $(y^1_i, y^2_i)$.
\end{proof}

\begin{lemma}  
[\cite{BovaChenValeriote13-generic}]
\label{lemma:non-cm-info}
Let $\bfB$ be a finite relational structure such that
$\caV(\Alg(\bfB))$ is not congruence modular.
There exists a relational structure $\bfA$
defined on a signature including three binary relation symbols
$\alpha$, $\beta$, and $\gamma$
which are preserved by the operations of some finite algebra in
$\caV(\Alg(\bfB))$,
such that the following hold:
\begin{itemize}

\item There exists a finite set $\caP$ of pentagons
where for each $\bfP \in \caP$,
the domain $P$ of $\bfP$ is a subset of $A$, and
it holds that
$\alpha^{\bfP} = \alpha^{\bfA} \cap P^2$,
$\beta^{\bfP} = \beta^{\bfA} \cap P^2$, and
$\gamma^{\bfP} = \gamma^{\bfA} \cap P^2$.
Moreover, there exists $\bfP \in \caP$ such that
$\bbL(\bfP)$ is a non-trivial lattice.

\item For each $k \geq 1$, there exists a relation $D_k \subseteq A^k$
which is pp-definable over $\bfA$ such that
for any elements $a_1, \ldots, a_k \in A$,
the tuple $(a_1, \ldots, a_k)$ is in $D_k$ if and only if
there exists a $\bfP \in \caP$ such that
all of the elements $a_1, \ldots, a_k$ are contained in the domain
$P$ of $\bfP$.
\end{itemize}
\end{lemma}

\begin{theorem}\label{thm:non-cm}
Let $\bfA$ be a relational structure
satisfying the conditions described in
Lemma~\ref{lemma:non-cm-info},
and let $\bfP \in \caP$ be a non-trivial pentagon
whose domain is not contained in that of any other pentagon
in $\caP$.
There exists a linear reduction from
$\peval(\bfP)$ to
$\eCSP(\bfA)$.
\end{theorem}

Note that this theorem makes use of a construction
from~\cite[Theorem 7]{BovaChenValeriote13-generic}
and shares elements in common with the proof of
\cite[Theorem 13]{Chen:2015wh}.

\begin{proof}
Let $(\phi(x_1, \ldots, x_m, y_1, \ldots, y_{\ell}), \bw), h$
denote the input to the first problem.
Let $x_{m+1}, \ldots, x_{m'}$ denote the quantified variables
of the first sort in $\phi$,
and let $y_{\ell+1}, \ldots, y_{\ell'}$ denote the
quantified variables of the second sort in $\phi$.
We use the translation of \cite[Theorem 7]{BovaChenValeriote13-generic}
to obtain a formula
$\phi'(x'_1, \ldots, x'_{m}, y'_1, \ldots, y'_{\ell}, v'_1, \ldots, v'_{|P|})$,
but instead of adding the conjunct of the form $\Delta_{\cdot}(\ldots)$,
we add the conjunct
$D_{m'+\ell'+|P|}(x'_1, \ldots, x'_{m'}, y'_1, \ldots, y'_{\ell'}, v'_1, \ldots, v'_{|P|})$
where $|P|$ denotes the size of the domain $P$ of $\bfP$
and the $v_i$ are fresh variables.
Let $b^*$ be a fixed element of $B_{\bfP}$,
and let $c^*$ be a fixed element of $C_{\bfP}$.
Define the assignment $h'$ as follows:
$h'(x'_i) = (h(x_i), c^*)$ for each $x_i$,
$h'(y'_i) = (b^*, h(y_i))$ for each $y_i$,
and let $h(v'_1), \ldots, h(v'_{|P|})$ be an enumeration of the
elements of $P$.
Set $\bw'$ so that $\bw'(v'_i) = 1/(|P|+1)$ for each variable $v'_i$,
and so that $\bw'(u) = \bw(u) / (|P|+1)$ for each other variable $u$
of $\phi'$.
It is straightforward to verify that the reduction
that outputs
$(\phi', \bw'), h'$ is correct.
\end{proof}

Combining Lemmas~\ref{lem:non-cm-reduction-1},~\ref{lem:non-cm-reduction-2}, and~\ref{lem:non-cm-reduction-3} and Theorem~\ref{thm:non-cm}, using Lemma~\ref{lem:reduction-for-ecsps}, we get the following.
\begin{theorem}\label{the:non-cm-ecsp}
  Let $\bfA$ be a relational structure such that $\caV(\Alg(\bfA))$ is not congruence modular.
  Then, $\eCSP(\bfA)$ is not sublinear-query testable with one-sided error.
\end{theorem}
From the argument at the beginning of this section, we obtain the following as a corollary.
\begin{theorem}\label{the:non-cm-ecsp-2}
  Let $\bfA$ be a relational structure that has not $(k+1)$-ary near unanimity polymorphism for any $k \ge 2$.
  Then, $\eCSP(\bfA)$ is not sublinear-query testable with one-sided error.
\end{theorem}



\bibliographystyle{abbrv}
\bibliography{main}

\end{document}